\documentclass{article}
\usepackage{amssymb}
\usepackage{amsmath}
\usepackage{amsthm}
\usepackage{color}
\usepackage[colorlinks]{hyperref}
\usepackage[a4paper]{geometry}
\usepackage{xurl}
\usepackage{rotating}
\usepackage{makecell}
\usepackage{array}
\usepackage{longtable}
\providecommand{\lcdGreedyTableRotationAngle}{-90}

\newtheorem{theorem}{Theorem}[section]

\newtheorem{corollary}[theorem]{Corollary}

\theoremstyle{definition}

\newtheorem{example}{Example}[section]

\theoremstyle{remark}

\title{New binary optimal LCD codes using heuristic embedding}
\author{Haeun Lim\thanks{haeunlim@sogang.ac.kr, Department of Mathematics and Institute for Mathematical and Data Sciences, Sogang University, Seoul, Korea}, Junmin An\thanks{junmin0518@sogang.ac.kr, Department of Mathematics and Institute for Mathematical and Data Sciences, Sogang University, Seoul, Korea}, Jon-Lark Kim\thanks{jlkim@sogang.ac.kr, Department of Mathematics and Institute for Mathematical and Data Sciences, Sogang University, Seoul, Korea}}
\date{}

\begin{document}
\maketitle

\begin{abstract}
    In this paper, we investigate the construction of binary optimal LCD codes through short LCD embeddings. For this purpose, we design heuristic frameworks based on a greedy algorithm. We explore the search spaces of LCD embeddings using the fact that an invertible matrix together with an arbitrary matrix yields an LCD embedding. We therefore use elementary row operations on the invertible block and single entry-flips on the arbitrary block as local moves in a greedy algorithm. Using this method, we have found $14$ optimal new LCD codes with dimensions 7 and 8 for lengths from 55 to 201.
\end{abstract}

\section{Introduction}
In 1992, J. L. Massey~\cite{M-1992} defined LCD codes as linear codes with complementary duals. He showed that the nearest-codeword decoding problem for LCD codes can be reduced to the problem of finding the nearest codeword in a code $\mathcal{C}$ given a codeword in its dual. Also, LCD codes have been the subject of wide interest since Carlet and Guilley~\cite{CG-2015} investigated their application against side-channel attacks and fault injection attack.

Since finding optimal linear codes is a central problem in coding theory, the same problem is also fundamental for LCD codes. There has been extensive research in determining the largest minimum distances of LCD codes. Galvez et al.~\cite{GKLRW-2018} determined the largest minimum distances of binary LCD codes for $n\le 12$. For binary LCD codes with $n\le 50$, partial results for various dimensions have been established in~\cite{AH-2020,B-2021,HS-2019,IS-2023,LSL-2024,WLL-2024}. A standard approach to determine the largest minimum distance of LCD codes is to explicitly construct codes obtaining the best known bounds. Various construction methods have been proposed in~\cite{B-2021,HS-2019,IS-2023,K-2023,LSL-2024,WLL-2024}, and these methods have produced new LCD codes that attain the best known upper bounds. In particular, Araya et al.~\cite{AHS-2021} provided a classification of binary optimal LCD codes of dimensions three and four. Also, Liu et al.~\cite{LLFS-2024} provided a partial classification of binary optimal LCD codes of dimension five, while Liu and Li~\cite{LL-2024} obtained partial results on binary LCD codes of dimension six.

More recently, embedding methods have been developed as a new approach to constructing optimal LCD codes. Kim et al.~\cite{KKL-2021} first proposed the shortest self-orthogonal embedding method, which constructs a given code into a self-orthogonal code by appending columns to its generator matrix. They provided an algorithm to embed a binary linear code of dimension three or four into a self-orthogonal code by appending the smallest number of columns. Kim and Choi~\cite{KC-2022} extended this result to codes of dimension five and six. The minimum number of columns required to embed a given binary code into a self-orthogonal code was completely determined by An et al.~\cite{AKKLW-arXiv1}. Later, An et al.~\cite{AHKL-2026} extended the shortest self-orthogonal embedding method to LCD codes. Although their shortest LCD embedding method produced several new LCD codes, they noted a limitation of their approach. Since the shortest LCD embeddings are parameterized by an arbitrary invertible matrix and an arbitrary matrix of prescribed sizes, the resulting search space is too large for exhaustive enumeration.

Recently, heuristic algorithms have been employed to construct new LCD codes in~\cite{WK-2026}. Motivated by this work, we apply a heuristic algorithm to the shortest LCD embedding problem to efficiently find, among the shortest LCD embeddings of a given code, one with the largest possible minimum distance. We use a greedy algorithm in which local moves are defined on the two blocks of a shortest LCD embedding, namely, the invertible-matrix block and the arbitrary-matrix block. Specifically, we use elementary row operations on the invertible-matrix block and randomly flipped a single entry of the arbitrary-matrix block. As a result, we obtained 19 optimal LCD codes with dimensions $6\le k\le 8$ for
lengths from $55$ to $225$, 14 of which are previously unknown. In addition, we obtained 159 LCD codes whose minimum distances are one below the bounds we have derived and 255 LCD codes whose minimum distances are two below those bounds.

This paper is organized as follows. Section 2 reviews basic notions from coding theory and introduces shortest LCD embeddings of binary linear codes, followed by a brief review of greedy algorithms. Section 3 describes our upper-bound construction and computational methodology. We present our computational results in Section 4. We conclude our paper in Section 5.

\section{Preliminaries}
\subsection{Linear codes}
Let $\mathbb{F}_2$ be the binary field with two elements. A \textit{code} $\mathcal{C}$ of length $n$ over $\mathbb{F}_2$ is a subset of $\mathbb{F}_2^n$, and elements of $\mathcal{C}$ are called \textit{codewords}. A \textit{linear code} is a $k$-dimensional subspace of $\mathbb{F}_2^n$ and denoted as an $[n,k]$ code. For an $[n,k]$ linear code $\mathcal{C}$, a \textit{generator matrix} $G$ for $\mathcal{C}$ is a $k\times n$ matrix over $\mathbb{F}_2$ whose rows form a basis of $\mathcal{C}$. Two codes $\mathcal{C}_1$ and $\mathcal{C}_2$ over $\mathbb{F}_2$ are called \textit{equivalent} if there exists a permutation of columns $\sigma$ where $\sigma \mathcal{C}_1 = \mathcal{C}_2$. Any linear code $\mathcal{C}$ over $\mathbb{F}_2$ has a generator matrix $G$ of the form
\[
G = [I_k~|~A]
\]
up to equivalence, where $I_k$ is the $k \times k$ identity matrix. Such a generator matrix is called the \textit{standard generator matrix} for $\mathcal{C}$. For an $[n,k]$ linear code $\mathcal{C}$, let $S$ be a subset of the set of coordinate positions of $\mathcal{C}$. The code obtained by removing the columns indexed by $S$ from $\mathcal{C}$ is called a \textit{punctured code} of $\mathcal{C}$.

The \textit{dual} $\mathcal{C}^{\perp}$ of a code $\mathcal{C}$ over $\mathbb{F}_2$ is defined as
\[
\mathcal{C}^\perp=\{\mathbf{x} \in \mathbb{F}_2^n~|~\mathbf{x} \cdot \mathbf{c} = 0  \textrm{ for all } \mathbf{c} \in \mathcal{C}\}
\]
where $\cdot$ is the standard dot product. The hull of a code $\mathcal{C}$ over $\mathbb{F}_2$ is defined as
\[
\operatorname{Hull}(\mathcal{C})=\mathcal{C}\cap \mathcal{C}^\perp.
\]
If $\operatorname{Hull}(\mathcal{C})=\{0\}$, then $\mathcal{C}$ is called an \textit{LCD code}. Let $G$ be a generator matrix of a linear code $\mathcal{C}$ over $\mathbb{F}_2$. It is known from~\cite{M-1992} that $\mathcal{C}$ is an LCD code if and only if $GG^T$ is invertible.

The \textit{Hamming weight} of a vector $\mathbf{x} \in \mathbb{F}_2^n$, denoted by $\operatorname{wt}(\mathbf{x})$, is the number of nonzero coordinates in $\mathbf{x}$. For two vectors $\mathbf{x}, \mathbf{y}$ in $\mathbb{F}_2^n$, the \textit{Hamming distance} between $\mathbf{x}$ and $\mathbf{y}$ is defined as
\[
d(\mathbf{x}, \mathbf{y})=\operatorname{wt}(\mathbf{x} - \mathbf{y}).
\]
The \textit{minimum (Hamming) distance} of a code $\mathcal{C}$, denoted by $d(\mathcal{C})$, is the minimum of the distances between any two distinct codewords in $\mathcal{C}$. An $[n,k]$ linear code $\mathcal{C}$ with minimum distance $d$ is denoted as an $[n,k,d]$ code.

Given $n$ and $k$, we denote by $d(n, k)$ the largest possible minimum distance over all linear $[n, k]$ codes. A linear $[n,k]$ code $\mathcal{C}$ with minimum distance $d(\mathcal{C})$ is called \textit{distance-optimal} if
\[
d(\mathcal{C})=d(n,k).
\]
Likewise, $d_{LCD}(n,k)$ denotes the largest possible minimum distance over all LCD $[n,k]$ codes. An LCD $[n,k]$ code $\mathcal{C}$ is called an \textit{optimal LCD code} if
\[
d(\mathcal{C})=d_{LCD}(n,k).
\]

\subsection{Shortest LCD embeddings of linear codes}
Given an $[n,k]$ linear code $\mathcal{C}$ over $\mathbb{F}_2$, an \textit{LCD embedding} $\tilde{\mathcal{C}}$ of $\mathcal{C}$ is an $[n',k]$ LCD code over $\mathbb{F}_2$ such that $\mathcal{C}$ is a punctured code of $\tilde{\mathcal{C}}$. Let $m=n'-n$ and $\mathcal{I}=\{i_1, \ldots , i_m\}$ be coordinate positions of $\tilde{\mathcal{C}}$ such that puncturing $\tilde{\mathcal{C}}$ at the coordinates in $\mathcal{I}$ yields $\mathcal{C}$. Up to equivalence, we may assume that the positions in $\mathcal{I}$ are the last $m$ coordinate positions.

Let $G$ be a generator matrix of $\mathcal{C}$. Consider the code $\mathcal{C}'$ generated by
\[
G'=[G~|~G~|~I].
\]
Since
\[
G'(G')^T=[G~|~G~|~I][G~|~G~|~I]^T=GG^T+GG^T+I=I,
\]
the code $\mathcal{C}'$ is LCD. Thus, any linear code over $\mathbb{F}_2$ has an LCD embedding.

Among all LCD embeddings of $\mathcal{C}$, one having the minimum length is called a \textit{shortest LCD embedding} of $\mathcal{C}$. The following two theorems determine the length and structure of shortest LCD embeddings.

\begin{theorem}[\cite{AHKL-2026}]
    Let $\mathcal{C}$ be an $[n,k]$ code over $\mathbb{F}_2$ with $\ell=\dim\operatorname{Hull}(\mathcal{C})$. Then the length $n'$ of a shortest LCD embedding of $\mathcal{C}$ is $n+\ell$.
\end{theorem}

\begin{theorem}[\cite{AHKL-2026}]\label{shortest-lcd-embedding}
    Let $\mathcal{C}$ be an $[n,k]$ code over $\mathbb{F}_2$ with $\ell=\dim\operatorname{Hull}(\mathcal{C})$. Let
    \[
    G = \begin{bmatrix}
        H \\ A
    \end{bmatrix}
    \]
    be a generator matrix of $\mathcal{C}$ where $H$ is a generator matrix of $\operatorname{Hull}(\mathcal{C})$. For an invertible $\ell\times\ell$ matrix $D$ and a $(k-\ell)\times\ell$ matrix $B$, let \[
    G' = \begin{bmatrix}
        H & D \\ A & B
    \end{bmatrix}.
    \]
    Then, $G'$ generates a shortest LCD embedding of $\mathcal{C}$. Conversely, every shortest LCD embedding of $\mathcal{C}$ has a generator matrix of this form.
\end{theorem}
By Theorem~\ref{shortest-lcd-embedding}, every shortest LCD embedding of $\mathcal{C}$ can be obtained, up to equivalence, by choosing an invertible matrix $D$ and an arbitrary matrix $B$.

\subsection{Greedy algorithm}

A greedy algorithm is a single-solution-based search method. Starting from an initial solution, the algorithm evaluates neighboring solutions using a \textit{fitness function} and moves to a locally best neighbor whenever a neighbor has a better fitness value. Here, the fitness value measures how well a solution satisfies the objective of the problem. This process is repeated until no neighboring solution gives a better fitness value, at which point the algorithm reaches a local optimum.

The main advantage of a greedy algorithm is its simplicity and efficiency. Since the algorithm only accepts moves that improve the fitness value, the search proceeds monotonically toward better solutions with respect to the chosen fitness function. However, this local decision rule also gives rise to a limitation. Because the algorithm does not consider the long-term effects of current choices, it may become trapped in a local optimum and fail to find a globally optimal solution.

One common way to reduce this limitation is to use a multi-start strategy. In this approach, the greedy search is repeated several times from different initial solutions, and the best solution found among all runs is selected. This allows the algorithm to explore different regions of the solution space and reduces the dependence on a single initial solution.

\section{Methods}
In this section, we present our methodology: we derive upper bounds on the minimum distances of binary LCD codes, and describe the experimental setup.

\subsection{Bounds}\label{experiments-bound-LCD}

Araya et al.~\cite{AHS-2021} determined the exact values of $d_{LCD}(n,4)$ for every $n$. Moreover, Liu et al.~\cite{LLFS-2024} determined the exact values of $d_{LCD}(n,5)$ for every $n$. Therefore, we focus on the cases where $6 \leq k\leq 10$.

We derive upper bounds for LCD codes over $\mathbb{F}_2$ of dimensions $6 \leq k \leq 10$. First, Grassl's table~\cite{grassl} provides bounds for $d(n,k)$ and in turn gives an upper bound for $d_{LCD}(n,k)$. For given $n$ and $k$, the following results allow us to derive upper bounds on $d_{LCD}(n,k)$ from known values of $d_{LCD}(n',k')$ where $n \geq n'$ and $k \geq k'$.
\begin{theorem}[\cite{CMTQ-2019}]\label{CMTQ-2019}
    If $2 \leq k \leq n$, then
    \[
    d_{LCD}(n,k) \leq d_{LCD}(n, k-1).
    \]
\end{theorem}

\begin{theorem}[\cite{LSL-2024}]\label{LCD-finite}
    If $2 \leq k \leq n$, then \[
    d_{LCD}(n,k) \leq \max \{d_{LCD}(n-1, k-1), d_{LCD}(n-2, k-2)\}.
    \]
\end{theorem}

Combining the bounds from Grassl's table~\cite{grassl} with Theorem~\ref{CMTQ-2019} and Theorem~\ref{LCD-finite}, we construct upper bounds on $d_{LCD}(n,k)$ for LCD $[n,k]$ codes over $\mathbb{F}_2$ in the range $51 \leq n \leq 256$ and $6 \leq k \leq 10$. The derived bounds are provided below.

\begin{corollary}\label{recursive-upper-bound}
Let $d_{g}(n,k)$ denote the upper bound for binary $[n,k]$ linear codes in Grassl's table~\cite{grassl}. Define $U(n,4)=d_{LCD}(n,4)$ and $U(n,5)=d_{LCD}(n,5)$ with the exact values
determined in~\cite{AHS-2021,LLFS-2024}, respectively, and, for $k\geq6$,
define recursively
\[
U(n,k)=\min\Bigl\{d_{\mathrm{Grassl}}(n,k),\ U(n,k-1),\
\max\{U(n-1,k-1),U(n-2,k-2)\}\Bigr\}.
\]
Then, for $51\le n\le256$ and $6\le k\le10$,
\[
d_{\mathrm{LCD}}(n,k)\le U(n,k),
\]
where the values of $U(n,k)$ are given in Table~\ref{bound_table}.

\begin{proof}
The assertion follows by induction on $k$.  The initial values are exact,
and $d_{\mathrm{LCD}}(n,k)\le d_{g}(n,k)$ because every LCD
code is a linear code. Theorem~\ref{CMTQ-2019} gives the second candidate
in the minimum as
\[
d_{LCD}(n, k)\le d_{LCD}(n, k-1)\le U(n, k-1)
\]
for $6\le k\le 10$ and Theorem~\ref{LCD-finite} gives the third as
\[
\begin{aligned}
d_{LCD}(n,k)
&\leq \max\{d_{LCD}(n-1,k-1),
             d_{LCD}(n-2,k-2)\}\\
&\leq \max\{U(n-1,k-1),U(n-2,k-2)\}.
\end{aligned}
\]
Thus each of the three entries in the defining minimum is an upper bound
on $d_{LCD}(n,k)$, so their minimum is also an upper bound.
\end{proof}
\begingroup
\setlength{\tabcolsep}{1pt}
\setlength{\LTleft}{\fill}
\setlength{\LTright}{\fill}
\begin{longtable}{>{\footnotesize\upshape\centering\arraybackslash}p{2.4em}|*{5}{>{\footnotesize\upshape\centering\arraybackslash}p{1.55em}}||>{\footnotesize\upshape\centering\arraybackslash}p{2.4em}|*{5}{>{\footnotesize\upshape\centering\arraybackslash}p{1.55em}}||>{\footnotesize\upshape\centering\arraybackslash}p{2.4em}|*{5}{>{\footnotesize\upshape\centering\arraybackslash}p{1.55em}}}
\caption{Upper bound on the minimum distance of LCD codes of lengths $51\le n\le 256$ and dimensions $6\le k\le 10$}
\label{bound_table}\\
\hline
$\mathrm{n}\backslash\mathrm{k}$ & $6$ & $7$ & $8$ & $9$ & $10$ & $\mathrm{n}\backslash\mathrm{k}$ & $6$ & $7$ & $8$ & $9$ & $10$ & $\mathrm{n}\backslash\mathrm{k}$ & $6$ & $7$ & $8$ & $9$ & $10$ \\
\hline
\endfirsthead
\hline
$\mathrm{n}\backslash\mathrm{k}$ & $6$ & $7$ & $8$ & $9$ & $10$ & $\mathrm{n}\backslash\mathrm{k}$ & $6$ & $7$ & $8$ & $9$ & $10$ & $\mathrm{n}\backslash\mathrm{k}$ & $6$ & $7$ & $8$ & $9$ & $10$ \\
\hline
\endhead
\hline
\endfoot
\hline
\endlastfoot
51 & 24 & 24 & 23 & 22 & 21 & 120 & 60 & 60 & 58 & 56 & 56 & 189 & 96 & 94 & 94 & 92 & 91 \\
52 & 25 & 24 & 24 & 22 & 22 & 121 & 60 & 60 & 58 & 57 & 56 & 190 & 96 & 95 & 94 & 92 & 92 \\
53 & 26 & 24 & 24 & 23 & 22 & 122 & 61 & 60 & 58 & 58 & 57 & 191 & 96 & 96 & 95 & 92 & 92 \\
54 & 26 & 24 & 24 & 24 & 23 & 123 & 62 & 61 & 59 & 58 & 58 & 192 & 96 & 96 & 96 & 93 & 92 \\
55 & 27 & 25 & 24 & 24 & 24 & 124 & 62 & 62 & 60 & 59 & 58 & 193 & 96 & 96 & 96 & 94 & 92 \\
56 & 27 & 26 & 24 & 24 & 24 & 125 & 63 & 62 & 61 & 60 & 58 & 194 & 96 & 96 & 96 & 94 & 93 \\
57 & 28 & 26 & 25 & 24 & 24 & 126 & 63 & 63 & 62 & 60 & 59 & 195 & 97 & 96 & 96 & 95 & 94 \\
58 & 28 & 27 & 26 & 25 & 24 & 127 & 64 & 63 & 63 & 60 & 60 & 196 & 98 & 96 & 96 & 96 & 94 \\
59 & 29 & 28 & 26 & 26 & 25 & 128 & 64 & 64 & 63 & 61 & 60 & 197 & 98 & 97 & 96 & 96 & 95 \\
60 & 29 & 28 & 27 & 26 & 26 & 129 & 64 & 64 & 64 & 62 & 60 & 198 & 99 & 98 & 97 & 96 & 96 \\
61 & 30 & 29 & 28 & 27 & 26 & 130 & 64 & 64 & 64 & 62 & 61 & 199 & 100 & 98 & 98 & 96 & 96 \\
62 & 30 & 30 & 28 & 28 & 27 & 131 & 64 & 64 & 64 & 62 & 62 & 200 & 100 & 99 & 98 & 96 & 96 \\
63 & 31 & 30 & 28 & 28 & 28 & 132 & 65 & 64 & 64 & 63 & 62 & 201 & 100 & 100 & 98 & 97 & 96 \\
64 & 31 & 31 & 29 & 28 & 28 & 133 & 66 & 64 & 64 & 64 & 63 & 202 & 101 & 100 & 99 & 98 & 96 \\
65 & 32 & 31 & 30 & 29 & 28 & 134 & 66 & 65 & 64 & 64 & 64 & 203 & 102 & 100 & 100 & 98 & 97 \\
66 & 32 & 32 & 30 & 30 & 28 & 135 & 67 & 66 & 65 & 64 & 64 & 204 & 102 & 101 & 100 & 99 & 98 \\
67 & 32 & 32 & 31 & 30 & 29 & 136 & 68 & 66 & 66 & 64 & 64 & 205 & 103 & 102 & 100 & 100 & 98 \\
68 & 32 & 32 & 32 & 30 & 30 & 137 & 68 & 67 & 66 & 64 & 64 & 206 & 104 & 102 & 101 & 100 & 99 \\
69 & 33 & 32 & 32 & 31 & 30 & 138 & 68 & 68 & 66 & 65 & 64 & 207 & 104 & 103 & 102 & 100 & 100 \\
70 & 34 & 33 & 32 & 32 & 31 & 139 & 69 & 68 & 67 & 66 & 65 & 208 & 104 & 104 & 102 & 101 & 100 \\
71 & 34 & 34 & 32 & 32 & 32 & 140 & 70 & 68 & 68 & 66 & 66 & 209 & 104 & 104 & 103 & 102 & 100 \\
72 & 35 & 34 & 32 & 32 & 32 & 141 & 70 & 69 & 68 & 67 & 66 & 210 & 105 & 104 & 104 & 102 & 101 \\
73 & 36 & 34 & 33 & 32 & 32 & 142 & 71 & 70 & 68 & 68 & 67 & 211 & 106 & 104 & 104 & 103 & 102 \\
74 & 36 & 35 & 34 & 33 & 32 & 143 & 72 & 70 & 69 & 68 & 68 & 212 & 106 & 105 & 104 & 104 & 102 \\
75 & 36 & 36 & 34 & 34 & 33 & 144 & 72 & 71 & 70 & 68 & 68 & 213 & 107 & 106 & 105 & 104 & 103 \\
76 & 37 & 36 & 35 & 34 & 34 & 145 & 72 & 72 & 70 & 69 & 68 & 214 & 108 & 106 & 106 & 104 & 104 \\
77 & 38 & 36 & 36 & 35 & 34 & 146 & 72 & 72 & 71 & 70 & 69 & 215 & 108 & 107 & 106 & 104 & 104 \\
78 & 38 & 37 & 36 & 36 & 35 & 147 & 73 & 72 & 72 & 70 & 70 & 216 & 108 & 108 & 107 & 105 & 104 \\
79 & 39 & 38 & 36 & 36 & 36 & 148 & 74 & 72 & 72 & 71 & 70 & 217 & 109 & 108 & 108 & 106 & 104 \\
80 & 39 & 38 & 37 & 36 & 36 & 149 & 74 & 73 & 72 & 72 & 71 & 218 & 110 & 108 & 108 & 106 & 105 \\
81 & 40 & 39 & 38 & 37 & 36 & 150 & 75 & 74 & 72 & 72 & 72 & 219 & 110 & 109 & 108 & 107 & 106 \\
82 & 40 & 40 & 38 & 38 & 36 & 151 & 76 & 74 & 73 & 72 & 72 & 220 & 111 & 110 & 109 & 108 & 106 \\
83 & 40 & 40 & 39 & 38 & 37 & 152 & 76 & 75 & 74 & 72 & 72 & 221 & 112 & 110 & 110 & 108 & 107 \\
84 & 41 & 40 & 40 & 39 & 38 & 153 & 76 & 76 & 74 & 73 & 72 & 222 & 112 & 111 & 110 & 108 & 108 \\
85 & 42 & 40 & 40 & 40 & 38 & 154 & 77 & 76 & 75 & 74 & 73 & 223 & 112 & 112 & 111 & 109 & 108 \\
86 & 42 & 41 & 40 & 40 & 39 & 155 & 78 & 76 & 76 & 74 & 74 & 224 & 112 & 112 & 112 & 110 & 108 \\
87 & 43 & 42 & 40 & 40 & 40 & 156 & 78 & 77 & 76 & 75 & 74 & 225 & 112 & 112 & 112 & 110 & 109 \\
88 & 44 & 42 & 41 & 40 & 40 & 157 & 79 & 78 & 76 & 76 & 75 & 226 & 113 & 112 & 112 & 110 & 110 \\
89 & 44 & 43 & 42 & 41 & 40 & 158 & 80 & 78 & 77 & 76 & 76 & 227 & 114 & 112 & 112 & 111 & 110 \\
90 & 44 & 44 & 42 & 42 & 40 & 159 & 80 & 79 & 78 & 76 & 76 & 228 & 114 & 113 & 112 & 112 & 111 \\
91 & 45 & 44 & 43 & 42 & 41 & 160 & 80 & 80 & 78 & 77 & 76 & 229 & 115 & 114 & 113 & 112 & 112 \\
92 & 46 & 45 & 44 & 43 & 42 & 161 & 80 & 80 & 79 & 78 & 77 & 230 & 116 & 114 & 114 & 112 & 112 \\
93 & 46 & 46 & 44 & 44 & 42 & 162 & 80 & 80 & 80 & 78 & 78 & 231 & 116 & 115 & 114 & 112 & 112 \\
94 & 47 & 46 & 44 & 44 & 43 & 163 & 81 & 80 & 80 & 79 & 78 & 232 & 116 & 116 & 115 & 113 & 112 \\
95 & 47 & 47 & 45 & 44 & 44 & 164 & 82 & 80 & 80 & 80 & 78 & 233 & 117 & 116 & 116 & 114 & 112 \\
96 & 48 & 47 & 46 & 44 & 44 & 165 & 82 & 81 & 80 & 80 & 79 & 234 & 118 & 116 & 116 & 114 & 113 \\
97 & 48 & 48 & 46 & 45 & 44 & 166 & 83 & 82 & 80 & 80 & 80 & 235 & 118 & 117 & 116 & 115 & 114 \\
98 & 48 & 48 & 47 & 46 & 44 & 167 & 84 & 82 & 81 & 80 & 80 & 236 & 119 & 118 & 117 & 116 & 114 \\
99 & 48 & 48 & 48 & 46 & 45 & 168 & 84 & 83 & 82 & 80 & 80 & 237 & 120 & 118 & 118 & 116 & 115 \\
100 & 49 & 48 & 48 & 47 & 46 & 169 & 84 & 84 & 82 & 81 & 80 & 238 & 120 & 119 & 118 & 116 & 116 \\
101 & 50 & 49 & 48 & 48 & 46 & 170 & 85 & 84 & 83 & 82 & 81 & 239 & 120 & 120 & 119 & 116 & 116 \\
102 & 50 & 50 & 48 & 48 & 47 & 171 & 86 & 84 & 84 & 82 & 82 & 240 & 120 & 120 & 120 & 117 & 116 \\
103 & 51 & 50 & 48 & 48 & 48 & 172 & 86 & 85 & 84 & 83 & 82 & 241 & 121 & 120 & 120 & 118 & 117 \\
104 & 52 & 51 & 49 & 48 & 48 & 173 & 87 & 86 & 84 & 84 & 83 & 242 & 122 & 120 & 120 & 118 & 118 \\
105 & 52 & 52 & 50 & 48 & 48 & 174 & 88 & 86 & 85 & 84 & 84 & 243 & 122 & 121 & 120 & 119 & 118 \\
106 & 52 & 52 & 50 & 49 & 48 & 175 & 88 & 87 & 86 & 84 & 84 & 244 & 123 & 122 & 121 & 120 & 119 \\
107 & 53 & 52 & 51 & 50 & 49 & 176 & 88 & 88 & 86 & 85 & 84 & 245 & 124 & 122 & 122 & 120 & 120 \\
108 & 54 & 53 & 52 & 50 & 50 & 177 & 88 & 88 & 87 & 86 & 84 & 246 & 124 & 123 & 122 & 120 & 120 \\
109 & 54 & 54 & 52 & 51 & 50 & 178 & 89 & 88 & 88 & 86 & 85 & 247 & 124 & 124 & 123 & 121 & 120 \\
110 & 55 & 54 & 52 & 52 & 51 & 179 & 90 & 88 & 88 & 87 & 86 & 248 & 125 & 124 & 124 & 122 & 120 \\
111 & 55 & 55 & 53 & 52 & 52 & 180 & 90 & 89 & 88 & 88 & 86 & 249 & 126 & 124 & 124 & 122 & 121 \\
112 & 56 & 55 & 54 & 52 & 52 & 181 & 91 & 90 & 88 & 88 & 87 & 250 & 126 & 125 & 124 & 122 & 122 \\
113 & 56 & 56 & 54 & 53 & 52 & 182 & 92 & 90 & 89 & 88 & 88 & 251 & 127 & 126 & 125 & 123 & 122 \\
114 & 56 & 56 & 55 & 54 & 53 & 183 & 92 & 91 & 90 & 88 & 88 & 252 & 128 & 126 & 126 & 124 & 123 \\
115 & 57 & 56 & 56 & 54 & 54 & 184 & 92 & 92 & 90 & 89 & 88 & 253 & 128 & 127 & 126 & 125 & 124 \\
116 & 58 & 57 & 56 & 55 & 54 & 185 & 93 & 92 & 91 & 90 & 88 & 254 & 128 & 128 & 127 & 126 & 124 \\
117 & 58 & 58 & 56 & 56 & 55 & 186 & 94 & 92 & 92 & 90 & 89 & 255 & 128 & 128 & 128 & 127 & 124 \\
118 & 59 & 58 & 56 & 56 & 56 & 187 & 94 & 93 & 92 & 90 & 90 & 256 & 128 & 128 & 128 & 128 & 124 \\
119 & 60 & 59 & 57 & 56 & 56 & 188 & 95 & 94 & 93 & 91 & 90 &  &  &  &  &  &  \\
\end{longtable}
\endgroup
\end{corollary}

\begin{example}
The recursion can be strictly sharper than the bound for general linear codes.  For instance, Grassl's table gives
$d_{g}(51,8)=24$.  Starting from the exact dimension-five
values in~\cite{LLFS-2024} and applying the recursion successively gives
$U(50,7)=23$, $U(49,6)=23$, and $U(51,7)=24$.  Hence
\[
U(51,8)
=\min\{24,24,\max\{23,23\}\}=23.
\]
Therefore Corollary~\ref{recursive-upper-bound} yields
$d_{LCD}(51,8)\leq23$, improving the Grassl-table upper bound by
one for binary LCD $[51,8]$ codes.
\end{example}

\subsection{Methodology}

In this section, we describe the methodology used to construct shortest LCD embeddings. Given a generator matrix $G$ of a linear code $\mathcal{C}$ over $\mathbb{F}_2$, our goal is to find, among the shortest LCD embeddings generated by our search, one with the largest possible minimum distance.

We aim to obtain $[n, k, d]$ LCD codes with large minimum distance $d$ for parameters in the range
\[
    51\leq n\leq 256
    \qquad\text{and}\qquad
    6\leq k\leq 10.
\]
In the case of LCD embeddings, greedy search is feasible because the admissible operations are relatively simple. Let $\mathcal{C}$ be a code over $\mathbb{F}_2$ such that $\dim(\operatorname{Hull}(\mathcal{C}))=\ell$ with generator matrix
\[
G=\begin{bmatrix}
    H\\
    A
\end{bmatrix}
\]
where $H$ generates $\operatorname{Hull}(\mathcal{C})$. Let
\[
G'=\begin{bmatrix}
H & D\\
A & B
\end{bmatrix}
\]
for some $\ell\times \ell$ matrix $D$ and $(k-\ell)\times \ell$ matrix $B$. Then $G'$ generates a shortest LCD embedding of $\mathcal{C}$. By Theorem~\ref{shortest-lcd-embedding}, the only structural constraint is that the block $D$ must be invertible, while the block $B$ is arbitrary.
\begin{enumerate}
    \item For the matrix $B$, the implemented moves are entry-flip operations
    \[
        \text{entry-flip:}\quad B_{ij}\leftarrow B_{ij}+1,
        \qquad
        1\leq i\leq k-\ell,\quad 1\leq j\leq \ell
    \]
    where the addition is over $\mathbb{F}_2$.
    \item For the invertible matrix $D$, let $r_i(D)$ be the $i$-th row of $D$. Then the implemented moves are given by
    \begin{align*}
        \text{row-add:}\quad
        & r_i(D)\leftarrow r_i(D)+r_j(D),
        \qquad 1\leq i,j\leq \ell,\quad i\neq j,\\
        \text{row-swap:}\quad
        & r_i(D)\leftrightarrow r_j(D),
        \qquad\qquad~~~~~ 1\leq i<j\leq \ell.
    \end{align*}
    These elementary row operations preserve invertibility and generate all of $GL_\ell(\mathbb{F}_2)$.
\end{enumerate}

Next, we describe the experimental design. We performed Experiment 1 twice and Experiment 2 twice.
    \begin{enumerate}
    \item \textbf{Experiment 1.} We allow both matrices $B$ and $D$ to change during the same search process. At each step, a local move is chosen randomly from the set of admissible operations.

    \item \textbf{Experiment 2.} We fix one matrix and update the other until a local optimum is reached. More precisely, when $D$ is fixed, we iteratively apply entry-flip operations to $B$. When $B$ is fixed, we update $D$ by choosing between row-addition and row-swap operations.
    \end{enumerate}

Experiment 1: joint local search.
    \begin{enumerate}
        \item Let $D_0$ be an invertible matrix from $GL_\ell(\mathbb{F}_2)$ and $B_0$ be a $(k-\ell)\times \ell$ arbitrary matrix.
        \item Randomly choose an operation from row-addition, row-swap operations of $D_0$ and entry-flip operation from $B_0$.
        \item If the resulting state improves the current fitness value, then the move is accepted. Otherwise, the current state is retained.
        \item This procedure is repeated until no further improvement is obtained or until the maximum number of steps, set to $60$, is reached. When the search reaches a local optimum, we restart it from a new initial state, up to $15$ times.
    \end{enumerate}

Experiment 2: block-coordinate local search.
    \begin{enumerate}
        \item Let $B_0$ be an initial $(k-\ell) \times \ell$ matrix over $\mathbb{F}_2$ and $D_0$ be an invertible matrix in $GL_\ell(\mathbb{F}_2)$.
        \item Fix $B$ and update $D$ by applying row-addition and row-swap operations. This $D$-phase is repeated until a local optimum with respect to $D$ is reached or until $30$ local moves have been performed.
        \item With the updated $D$ fixed, update $B$ by applying entry-flip operations. This $B$-phase is repeated until a local optimum with respect to $B$ is reached or until $30$ local moves have been performed. We repeat the $D$-phase and the $B$-phase alternately until the fitness value no longer improves.
        \item When the alternating process reaches a local optimum, we restart it from a new initial state $(B_0,D_0)$, up to $15$ times.
    \end{enumerate}

To evaluate each shortest LCD embedding, we use a lexicographic fitness function. For a state $(B,D)$, let $\mathcal{C}_{B,D}$ be the shortest LCD embedding obtained from the corresponding generator matrix. Let
    \[
    d_{B,D}=d(\mathcal{C}_{B,D}),
    \]
    and let $A_i(\mathcal{C}_{B,D})$ denote the number of codewords of weight $i$ in $\mathcal{C}_{B,D}$. We define the fitness of $(B,D)$ by
    \[
    F(B,D)=\left(d_{B,D},-A_{d_{B,D}}(\mathcal{C}_{B,D})\right).
    \]
    Fitness values are compared lexicographically. Thus, a shortest LCD embedding with larger minimum distance is preferred. If two shortest LCD embeddings have the same minimum distance, then the one with fewer minimum-weight codewords is preferred.

\section{Numerical results}
We conducted experiments on a total of $1019$ base codes obtained from the MAGMA BKLC database. After the embedding process, some distinct BKLC base codes produced LCD codes with identical $[n, k]$ parameters. To avoid counting such redundant cases, we kept only one representative for each parameter set, namely, the code with the largest minimum distance among those with the same parameters. After removing, for each parameter set, all but a representative with the largest minimum distance, $756$ representative codes remain. The results for these representative codes are as follows.
\begin{itemize}
    \item Bound attained: 19 codes.
    \item 1 below the bound: 159 codes.
    \item 2 below the bound: 255 codes.
    \item 3 below the bound: 219 codes.
    \item 4 below the bound: 76 codes.
    \item 5 below the bound: 23 codes.
    \item 6 below the bound: 4 codes.
    \item 7 below the bound: 1 code.
\end{itemize}
The generator matrices of the 19 LCD codes that attain the bound are given below, which are presented in hexadecimal format $(A=10, B=11, C=12, D=13, E=14, F=15)$ after appending them with zero columns to ensure that the length is a multiple of 4.

\begingroup
\setlength{\arraycolsep}{0pt}
\begin{itemize}
\item An LCD $[69,6,33]$ code obtained from a $[63,6,32]$ code, whose generator matrix is
{\tiny
\[
G_{\mathrm{LCD, 1}}=\left[
\begin{array}{cccccccccccccccccc}
9 & 6 & 6 & 9 & 6 & 9 & 9 & 6 & 9 & 6 & 6 & 9 & 6 & 9 & 9 & 7 & 0 & 0 \\
5 & 5 & 5 & 5 & 5 & 5 & 5 & 5 & A & A & A & A & A & A & A & A & 8 & 0 \\
3 & 3 & 3 & 3 & 3 & 3 & 3 & 3 & C & C & C & C & C & C & C & C & 4 & 0 \\
0 & F & 0 & F & 0 & F & 0 & F & F & 0 & F & 0 & F & 0 & F & 0 & 2 & 0 \\
0 & 0 & F & F & 0 & 0 & F & F & F & F & 0 & 0 & F & F & 0 & 0 & 1 & 0 \\
0 & 0 & 0 & 0 & F & F & F & F & F & F & F & F & 0 & 0 & 0 & 0 & 0 & 8
\end{array}
\right]_{16}.
\]
}

\item An LCD $[132,6,65]$ code obtained from a $[126,6,64]$ code, whose generator matrix is
{\tiny
\[
G_{\mathrm{LCD, 2}}=\left[
\begin{array}{ccccccccccccccccccccccccccccccccc}
C & 3 & 3 & C & 3 & C & C & 3 & 3 & C & C & 3 & C & 3 & 3 & C & C & 3 & 3 & C & 3 & C & C & 3 & 3 & C & C & 3 & C & 3 & 3 & E & 0 \\
3 & 3 & 3 & 3 & 3 & 3 & 3 & 3 & 3 & 3 & 3 & 3 & 3 & 3 & 3 & 3 & C & C & C & C & C & C & C & C & C & C & C & C & C & C & C & D & 0 \\
0 & F & 0 & F & 0 & F & 0 & F & 0 & F & 0 & F & 0 & F & 0 & F & F & 0 & F & 0 & F & 0 & F & 0 & F & 0 & F & 0 & F & 0 & F & 0 & 8 \\
0 & 0 & F & F & 0 & 0 & F & F & 0 & 0 & F & F & 0 & 0 & F & F & F & F & 0 & 0 & F & F & 0 & 0 & F & F & 0 & 0 & F & F & 0 & 0 & 4 \\
0 & 0 & 0 & 0 & F & F & F & F & 0 & 0 & 0 & 0 & F & F & F & F & F & F & F & F & 0 & 0 & 0 & 0 & F & F & F & F & 0 & 0 & 0 & 0 & 2 \\
0 & 0 & 0 & 0 & 0 & 0 & 0 & 0 & F & F & F & F & F & F & F & F & F & F & F & F & F & F & F & F & 0 & 0 & 0 & 0 & 0 & 0 & 0 & 0 & 1
\end{array}
\right]_{16}.
\]
}

\item An LCD $[162,6,80]$ code obtained from a $[157,6,79]$ code, whose generator matrix is
{\tiny
\[
G_{\mathrm{LCD, 3}}=\left[
\begin{array}{ccccccccccccccccccccccccccccccccccccccccc}
9 & 6 & 6 & 9 & 6 & 9 & 9 & 7 & 8 & 6 & 7 & 8 & 7 & 9 & 8 & 6 & 7 & 9 & 8 & 7 & 8 & 6 & 7 & 9 & 8 & 6 & 7 & 8 & 7 & 9 & 8 & 6 & 7 & 9 & 8 & 7 & 8 & 6 & 7 & A & 8 \\
5 & 5 & 5 & 5 & A & A & A & A & 6 & 6 & 6 & 6 & 6 & 6 & 6 & 7 & 9 & 9 & 9 & 9 & 9 & 9 & 9 & 8 & 6 & 6 & 6 & 6 & 6 & 6 & 6 & 7 & 9 & 9 & 9 & 9 & 9 & 9 & 9 & F & 0 \\
3 & 3 & 3 & 3 & C & C & C & C & 1 & E & 1 & E & 1 & E & 1 & F & E & 1 & E & 1 & E & 1 & E & 0 & 1 & E & 1 & E & 1 & E & 1 & F & E & 1 & E & 1 & E & 1 & E & 3 & 4 \\
0 & F & 0 & F & F & 0 & F & 0 & 0 & 1 & F & E & 0 & 1 & F & F & F & E & 0 & 1 & F & E & 0 & 0 & 0 & 1 & F & E & 0 & 1 & F & F & F & E & 0 & 1 & F & E & 0 & 2 & 4 \\
0 & 0 & F & F & F & F & 0 & 0 & 0 & 0 & 0 & 1 & F & F & F & F & F & F & F & E & 0 & 0 & 0 & 0 & 0 & 0 & 0 & 1 & F & F & F & F & F & F & F & E & 0 & 0 & 0 & 6 & 8 \\
5 & 5 & 5 & 5 & 5 & 5 & 5 & 4 & 6 & 6 & 6 & 6 & 6 & 6 & 6 & 6 & 6 & 6 & 6 & 6 & 6 & 6 & 6 & 7 & 9 & 9 & 9 & 9 & 9 & 9 & 9 & 9 & 9 & 9 & 9 & 9 & 9 & 9 & 9 & D & 0
\end{array}
\right]_{16}.
\]
}

\item An LCD $[195,6,97]$ code obtained from a $[189,6,96]$ code, whose generator matrix is
{\tiny
\[
G_{\mathrm{LCD, 4}}=\left[
\begin{array}{ccccccccccccccccccccccccccccccccccccccccccccccccc}
9 & 6 & 6 & 9 & 6 & 9 & 9 & 6 & 6 & 9 & 9 & 6 & 9 & 6 & 6 & B & 0 & C & F & 0 & F & 3 & 0 & C & F & 3 & 0 & F & 0 & C & F & 3 & 0 & C & F & 0 & F & 3 & 0 & C & F & 3 & 0 & F & 0 & C & F & C & 0 \\
5 & 5 & 5 & 5 & 5 & 5 & 5 & 5 & 5 & 5 & 5 & 5 & 5 & 5 & 5 & 4 & C & C & C & C & C & C & C & C & C & C & C & C & C & C & C & F & 3 & 3 & 3 & 3 & 3 & 3 & 3 & 3 & 3 & 3 & 3 & 3 & 3 & 3 & 3 & A & 0 \\
3 & 3 & 3 & 3 & 3 & 3 & 3 & 3 & C & C & C & C & C & C & C & C & 3 & C & 3 & C & 3 & C & 3 & C & 3 & C & 3 & C & 3 & C & 3 & F & C & 3 & C & 3 & C & 3 & C & 3 & C & 3 & C & 3 & C & 3 & C & 1 & 0 \\
0 & F & 0 & F & 0 & F & 0 & F & F & 0 & F & 0 & F & 0 & F & 0 & 0 & 3 & F & C & 0 & 3 & F & C & 0 & 3 & F & C & 0 & 3 & F & F & F & C & 0 & 3 & F & C & 0 & 3 & F & C & 0 & 3 & F & C & 0 & 0 & 8 \\
0 & 0 & F & F & 0 & 0 & F & F & F & F & 0 & 0 & F & F & 0 & 0 & 0 & 0 & 0 & 3 & F & F & F & C & 0 & 0 & 0 & 3 & F & F & F & F & F & F & F & C & 0 & 0 & 0 & 3 & F & F & F & C & 0 & 0 & 0 & 0 & 4 \\
0 & 0 & 0 & 0 & F & F & F & F & F & F & F & F & 0 & 0 & 0 & 0 & 0 & 0 & 0 & 0 & 0 & 0 & 0 & 3 & F & F & F & F & F & F & F & F & F & F & F & F & F & F & F & C & 0 & 0 & 0 & 0 & 0 & 0 & 0 & 0 & 2
\end{array}
\right]_{16}.
\]
}

\item An LCD $[225,6,112]$ code obtained from a $[220,6,111]$ code, whose generator matrix is
{\tiny
\[
G_{\mathrm{LCD, 5}}=\left[
\begin{array}{ccccccccccccccccccccccccccccccccccccccccccccccccccccccccc}
9 & 6 & 6 & 9 & 6 & 9 & 9 & 7 & 2 & C & D & 2 & D & 3 & 2 & C & D & 3 & 2 & D & 2 & C & D & 6 & 1 & 9 & E & 1 & E & 6 & 1 & 9 & E & 6 & 1 & E & 1 & 9 & E & 6 & 1 & 9 & E & 1 & E & 6 & 1 & 9 & E & 6 & 1 & E & 1 & 9 & F & 8 & 0 \\
5 & 5 & 5 & 5 & A & A & A & A & A & A & A & B & 5 & 5 & 5 & 5 & 5 & 5 & 5 & 4 & A & A & A & 9 & 9 & 9 & 9 & 9 & 9 & 9 & 9 & E & 6 & 6 & 6 & 6 & 6 & 6 & 6 & 1 & 9 & 9 & 9 & 9 & 9 & 9 & 9 & E & 6 & 6 & 6 & 6 & 6 & 6 & 7 & C & 8 \\
3 & 3 & 3 & 3 & C & C & C & C & 6 & 6 & 6 & 7 & 9 & 9 & 9 & 8 & 6 & 6 & 6 & 7 & 9 & 9 & 9 & 8 & 7 & 8 & 7 & 8 & 7 & 8 & 7 & F & 8 & 7 & 8 & 7 & 8 & 7 & 8 & 0 & 7 & 8 & 7 & 8 & 7 & 8 & 7 & F & 8 & 7 & 8 & 7 & 8 & 7 & 8 & 1 & 8 \\
0 & F & 0 & F & F & 0 & F & 0 & 1 & E & 1 & F & E & 1 & E & 0 & 1 & E & 1 & F & E & 1 & E & 0 & 0 & 7 & F & 8 & 0 & 7 & F & F & F & 8 & 0 & 7 & F & 8 & 0 & 0 & 0 & 7 & F & 8 & 0 & 7 & F & F & F & 8 & 0 & 7 & F & 8 & 0 & 7 & 0 \\
0 & 0 & F & F & F & F & 0 & 0 & 0 & 1 & F & F & F & E & 0 & 0 & 0 & 1 & F & F & F & E & 0 & 0 & 0 & 0 & 0 & 7 & F & F & F & F & F & F & F & 8 & 0 & 0 & 0 & 0 & 0 & 0 & 0 & 7 & F & F & F & F & F & F & F & 8 & 0 & 0 & 0 & A & 8 \\
5 & 5 & 5 & 5 & 5 & 5 & 5 & 4 & A & A & A & A & A & A & A & A & A & A & A & A & A & A & A & 9 & 9 & 9 & 9 & 9 & 9 & 9 & 9 & 9 & 9 & 9 & 9 & 9 & 9 & 9 & 9 & E & 6 & 6 & 6 & 6 & 6 & 6 & 6 & 6 & 6 & 6 & 6 & 6 & 6 & 6 & 7 & 4 & 0
\end{array}
\right]_{16}.
\]
}

\item An LCD $[55,7,25]$ code obtained from a $[49,7,23]$ code, whose generator matrix is
{\tiny
\[
G_{\mathrm{LCD, 6}}=\left[
\begin{array}{cccccccccccccc}
8 & 1 & 9 & 7 & 7 & 5 & 6 & 2 & 9 & A & D & 4 & F & E \\
4 & 2 & F & F & 0 & 9 & C & 8 & 9 & 4 & 5 & B & B & A \\
2 & 3 & 4 & B & 3 & 7 & 9 & D & 9 & 3 & 1 & C & 7 & A \\
1 & 1 & A & 5 & 9 & B & C & E & C & 9 & 8 & E & 3 & 6 \\
0 & 8 & D & 2 & C & D & E & 7 & 6 & 4 & C & 7 & 5 & 4 \\
0 & 4 & 6 & 9 & 6 & 6 & F & 3 & B & 2 & 6 & 3 & E & 4 \\
4 & 1 & D & 1 & E & 3 & 0 & D & A & 1 & F & 2 & C & 6
\end{array}
\right]_{16}.
\]
}

\item An LCD $[57,7,26]$ code obtained from a $[56,7,26]$ code, whose generator matrix is
{\tiny
\[
G_{\mathrm{LCD, 7}}=\left[
\begin{array}{ccccccccccccccc}
6 & 3 & 3 & 6 & B & 6 & 7 & 6 & 6 & D & A & 6 & 5 & E & 8 \\
9 & B & 3 & 9 & 5 & 6 & 4 & 6 & 9 & 3 & A & 5 & A & C & 8 \\
5 & A & A & 5 & 2 & 5 & 5 & 5 & 4 & B & 5 & 5 & 7 & 9 & 0 \\
3 & 9 & 9 & 3 & 6 & C & D & 3 & 2 & 6 & C & C & D & 8 & 0 \\
0 & 7 & 8 & F & 1 & C & 3 & 0 & E & 1 & C & 3 & C & 7 & 0 \\
0 & 0 & 7 & F & 0 & 3 & F & 0 & 1 & F & C & 0 & 3 & F & 0 \\
0 & 0 & 0 & 0 & 0 & 0 & 0 & F & F & F & F & F & F & F & 0
\end{array}
\right]_{16}.
\]
}

\item An LCD $[69,7,32]$ code obtained from a $[63,7,31]$ code, whose generator matrix is
{\tiny
\[
G_{\mathrm{LCD, 8}}=\left[
\begin{array}{cccccccccccccccccc}
9 & 6 & 6 & 9 & 6 & 9 & 9 & 6 & 9 & 6 & 6 & 9 & 6 & 9 & 9 & 6 & 1 & 0 \\
5 & 5 & 5 & 5 & 5 & 5 & 5 & 5 & A & A & A & A & A & A & A & A & C & 8 \\
3 & 3 & 3 & 3 & 3 & 3 & 3 & 3 & C & C & C & C & C & C & C & D & A & 0 \\
0 & F & 0 & F & 0 & F & 0 & F & F & 0 & F & 0 & F & 0 & F & 1 & 6 & 0 \\
0 & 0 & F & F & 0 & 0 & F & F & F & F & 0 & 0 & F & F & 0 & 0 & 2 & 0 \\
0 & 0 & 0 & 0 & F & F & F & F & F & F & F & F & 0 & 0 & 0 & 1 & 3 & 0 \\
9 & 6 & 6 & 9 & 6 & 9 & 9 & 6 & 6 & 9 & 9 & 6 & 9 & 6 & 6 & 9 & 1 & 8
\end{array}
\right]_{16}.
\]
}

\item An LCD $[72,7,34]$ code obtained from a $[71,7,34]$ code, whose generator matrix is
{\tiny
\[
G_{\mathrm{LCD, 9}}=\left[
\begin{array}{cccccccccccccccccc}
F & F & F & F & F & F & F & F & F & F & F & F & F & F & F & F & 0 & 1 \\
8 & 1 & B & 8 & E & 9 & C & 3 & 3 & A & D & 1 & 5 & C & 9 & 7 & F & C \\
4 & 1 & 6 & 4 & 9 & D & 2 & 2 & A & 7 & B & 9 & F & 2 & D & C & 8 & 2 \\
2 & 0 & F & 6 & 6 & E & 3 & 0 & D & E & 9 & 4 & 7 & 5 & 2 & E & 4 & 2 \\
1 & 0 & 2 & 3 & 4 & 7 & B & 9 & 6 & F & 4 & A & 9 & C & D & E & 2 & 2 \\
0 & 8 & 0 & D & F & 7 & 1 & C & 3 & 3 & C & D & 6 & A & 2 & F & 1 & 2 \\
0 & 4 & 1 & 2 & 8 & F & C & E & 9 & D & A & 6 & 3 & B & 5 & 7 & 0 & A
\end{array}
\right]_{16}.
\]
}

\item An LCD $[73,7,34]$ code obtained from a $[72,7,34]$ code, whose generator matrix is
{\tiny
\[
G_{\mathrm{LCD, 10}}=\left[
\begin{array}{ccccccccccccccccccc}
F & F & F & F & F & F & F & F & F & F & F & F & F & F & F & F & 0 & 0 & 8 \\
8 & 1 & B & 8 & E & 9 & C & 3 & 3 & A & D & 1 & 5 & C & 9 & 7 & F & C & 0 \\
4 & 1 & 6 & 4 & 9 & D & 2 & 2 & A & 7 & B & 9 & F & 2 & D & C & 8 & 2 & 0 \\
2 & 0 & F & 6 & 6 & E & 3 & 0 & D & E & 9 & 4 & 7 & 5 & 2 & E & 4 & 2 & 0 \\
1 & 0 & 2 & 3 & 4 & 7 & B & 9 & 6 & F & 4 & A & 9 & C & D & E & 2 & 2 & 0 \\
0 & 8 & 0 & D & F & 7 & 1 & C & 3 & 3 & C & D & 6 & A & 2 & F & 1 & 2 & 0 \\
0 & 4 & 1 & 2 & 8 & F & C & E & 9 & D & A & 6 & 3 & B & 5 & 7 & 0 & A & 0
\end{array}
\right]_{16}.
\]
}

\item An LCD $[85,7,40]$ code obtained from a $[84,7,40]$ code, whose generator matrix is
{\tiny
\[
G_{\mathrm{LCD, 11}}=\left[
\begin{array}{cccccccccccccccccccccc}
F & E & 0 & 0 & 0 & 0 & 0 & F & F & F & F & F & 8 & 0 & F & E & 0 & 3 & F & 8 & 0 & 8 \\
8 & 1 & 8 & 3 & 0 & 4 & 8 & B & 1 & A & 2 & A & 5 & C & F & 9 & 5 & B & 7 & 7 & E & 8 \\
4 & 0 & C & 1 & 8 & 2 & 4 & 5 & 8 & D & 1 & 5 & 2 & E & 7 & C & A & D & B & B & F & 8 \\
2 & 0 & 6 & 0 & C & 1 & 2 & 2 & C & 6 & 8 & A & 9 & 7 & 3 & F & 5 & 6 & D & D & F & 0 \\
1 & 0 & 3 & 0 & 6 & 0 & 9 & 1 & 6 & 3 & 6 & 5 & 4 & B & 9 & F & A & B & 6 & E & F & 0 \\
0 & 8 & 1 & 8 & 3 & 4 & 4 & 8 & B & 1 & 9 & 2 & E & 5 & C & E & D & 7 & B & 7 & 7 & 0 \\
0 & 4 & 0 & C & 1 & A & 2 & C & 4 & 8 & E & 9 & 7 & 2 & E & 7 & 6 & 9 & D & F & B & 0
\end{array}
\right]_{16}.
\]
}

\item An LCD $[134,7,65]$ code obtained from a $[127,7,64]$ code, whose generator matrix is
{\tiny
\[
G_{\mathrm{LCD, 12}}=\left[
\begin{array}{cccccccccccccccccccccccccccccccccc}
9 & 6 & 6 & 9 & 6 & 9 & 9 & 6 & 6 & 9 & 9 & 6 & 9 & 6 & 6 & 9 & 6 & 9 & 9 & 6 & 9 & 6 & 6 & 9 & 9 & 6 & 6 & 9 & 6 & 9 & 9 & 7 & 0 & 0 \\
5 & 5 & 5 & 5 & 5 & 5 & 5 & 5 & 5 & 5 & 5 & 5 & 5 & 5 & 5 & 5 & A & A & A & A & A & A & A & A & A & A & A & A & A & A & A & A & 8 & 0 \\
3 & 3 & 3 & 3 & 3 & 3 & 3 & 3 & 3 & 3 & 3 & 3 & 3 & 3 & 3 & 3 & C & C & C & C & C & C & C & C & C & C & C & C & C & C & C & C & 4 & 0 \\
0 & F & 0 & F & 0 & F & 0 & F & 0 & F & 0 & F & 0 & F & 0 & F & F & 0 & F & 0 & F & 0 & F & 0 & F & 0 & F & 0 & F & 0 & F & 0 & 2 & 0 \\
0 & 0 & F & F & 0 & 0 & F & F & 0 & 0 & F & F & 0 & 0 & F & F & F & F & 0 & 0 & F & F & 0 & 0 & F & F & 0 & 0 & F & F & 0 & 0 & 1 & 0 \\
0 & 0 & 0 & 0 & F & F & F & F & 0 & 0 & 0 & 0 & F & F & F & F & F & F & F & F & 0 & 0 & 0 & 0 & F & F & F & F & 0 & 0 & 0 & 0 & 0 & 8 \\
0 & 0 & 0 & 0 & 0 & 0 & 0 & 0 & F & F & F & F & F & F & F & F & F & F & F & F & F & F & F & F & 0 & 0 & 0 & 0 & 0 & 0 & 0 & 0 & 0 & 4
\end{array}
\right]_{16}.
\]
}

\item An LCD $[136,7,66]$ code obtained from a $[135,7,66]$ code, whose generator matrix is
{\tiny
\[
G_{\mathrm{LCD, 13}}=\left[
\begin{array}{cccccccccccccccccccccccccccccccccc}
7 & F & 0 & 2 & 0 & C & 2 & 8 & F & 2 & 2 & C & E & A & 7 & D & 0 & E & 2 & 4 & D & A & D & E & C & 6 & 9 & 7 & 7 & 3 & 2 & B & F & F \\
8 & 1 & 0 & 6 & 1 & 4 & 7 & 9 & 1 & 6 & 7 & 5 & 3 & E & 8 & 7 & 1 & 2 & 6 & D & 6 & F & 6 & 3 & 4 & B & B & 9 & 9 & 5 & 7 & F & F & 8 \\
4 & 1 & 8 & 5 & 1 & E & 4 & 5 & 9 & D & 4 & F & A & 1 & C & 4 & 9 & B & 5 & B & D & 8 & D & 2 & E & E & 6 & 5 & 5 & F & C & 1 & 0 & 4 \\
2 & 0 & C & 2 & 8 & F & 2 & 2 & C & E & A & 7 & D & 0 & E & 2 & 4 & D & A & D & E & C & 6 & 9 & 7 & 7 & 3 & 2 & A & F & E & 0 & 8 & 2 \\
1 & 0 & 6 & 1 & 4 & 7 & 9 & 1 & 6 & 7 & 5 & 3 & E & 8 & 7 & 1 & 2 & 6 & D & 6 & F & 6 & 3 & 4 & B & B & 9 & 9 & 5 & 7 & F & 0 & 4 & 2 \\
0 & 8 & 3 & 0 & A & 3 & C & 8 & B & 3 & A & 9 & F & 4 & 3 & 8 & 9 & 3 & 6 & B & 7 & B & 1 & A & 5 & D & C & C & A & B & F & 8 & 2 & 2 \\
0 & 4 & 1 & 8 & 5 & 1 & E & 4 & 5 & 9 & D & 4 & F & A & 1 & C & 4 & 9 & B & 5 & B & D & 8 & D & 2 & E & E & 6 & 5 & 5 & F & C & 1 & 2
\end{array}
\right]_{16}.
\]
}

\item An LCD $[199,7,98]$ code obtained from a $[198,7,98]$ code, whose generator matrix is
{\tiny
\[
G_{\mathrm{LCD, 14}}=\left[
\begin{array}{cccccccccccccccccccccccccccccccccccccccccccccccccc}
7 & F & 0 & 2 & 0 & C & 2 & 8 & F & 2 & 2 & C & E & A & 7 & D & 0 & E & 2 & 4 & D & A & D & E & C & 6 & 9 & 7 & 7 & 3 & 2 & B & F & F & F & F & F & F & F & F & F & F & F & F & F & F & F & E & 0 & 2 \\
8 & 1 & 0 & 6 & 1 & 4 & 7 & 9 & 1 & 6 & 7 & 5 & 3 & E & 8 & 7 & 1 & 2 & 6 & D & 6 & F & 6 & 3 & 4 & B & B & 9 & 9 & 5 & 7 & F & F & 8 & 2 & C & 4 & B & 7 & 5 & 0 & 5 & B & A & 2 & D & 8 & 6 & 0 & C \\
4 & 1 & 8 & 5 & 1 & E & 4 & 5 & 9 & D & 4 & F & A & 1 & C & 4 & 9 & B & 5 & B & D & 8 & D & 2 & E & E & 6 & 5 & 5 & F & C & 1 & 0 & 4 & A & 2 & 6 & 7 & 0 & C & 8 & F & E & 7 & 6 & B & 5 & 7 & F & 4 \\
2 & 0 & C & 2 & 8 & F & 2 & 2 & C & E & A & 7 & D & 0 & E & 2 & 4 & D & A & D & E & C & 6 & 9 & 7 & 7 & 3 & 2 & A & F & E & 0 & 8 & 2 & C & 9 & 3 & A & 4 & 5 & 4 & F & 7 & 3 & E & 5 & B & 9 & 0 & 4 \\
1 & 0 & 6 & 1 & 4 & 7 & 9 & 1 & 6 & 7 & 5 & 3 & E & 8 & 7 & 1 & 2 & 6 & D & 6 & F & 6 & 3 & 4 & B & B & 9 & 9 & 5 & 7 & F & 0 & 4 & 1 & E & C & D & C & 6 & 1 & B & D & 2 & 8 & E & A & 5 & C & 8 & 4 \\
0 & 8 & 3 & 0 & A & 3 & C & 8 & B & 3 & A & 9 & F & 4 & 3 & 8 & 9 & 3 & 6 & B & 7 & B & 1 & A & 5 & D & C & C & A & B & F & 8 & 2 & 0 & 4 & 6 & 8 & F & 7 & 2 & D & E & 9 & 5 & 3 & 9 & B & C & 4 & 4 \\
0 & 4 & 1 & 8 & 5 & 1 & E & 4 & 5 & 9 & D & 4 & F & A & 1 & C & 4 & 9 & B & 5 & B & D & 8 & D & 2 & E & E & 6 & 5 & 5 & F & C & 1 & 0 & 1 & B & E & E & 3 & 8 & 6 & 7 & 9 & A & D & 4 & 5 & E & 2 & 4
\end{array}
\right]_{16}.
\]
}

\item An LCD $[57,8,25]$ code obtained from a $[50,8,23]$ code, whose generator matrix is
{\tiny
\[
G_{\mathrm{LCD, 15}}=\left[
\begin{array}{ccccccccccccccc}
8 & 1 & 9 & 7 & 7 & 5 & 6 & 2 & 9 & A & D & 4 & B & 8 & 8 \\
4 & 0 & C & B & B & A & B & 1 & 4 & D & 6 & A & 5 & D & 8 \\
2 & 1 & 7 & F & 8 & 4 & E & 4 & 4 & A & 2 & D & C & 7 & 8 \\
1 & 1 & A & 5 & 9 & B & C & E & C & 9 & 8 & E & 0 & 7 & 0 \\
0 & 8 & D & 2 & C & D & E & 7 & 6 & 4 & C & 7 & 2 & 6 & 8 \\
0 & 4 & 6 & 9 & 6 & 6 & F & 3 & B & 2 & 6 & 3 & 8 & 8 & 0 \\
0 & 2 & 3 & 4 & B & 3 & 7 & 9 & D & 9 & 3 & 1 & D & C & 8 \\
8 & 0 & 8 & D & 2 & C & D & E & 7 & 6 & 4 & C & 6 & 0 & 8
\end{array}
\right]_{16}.
\]
}

\item An LCD $[134,8,64]$ code obtained from a $[127,8,63]$ code, whose generator matrix is
{\tiny
\[
G_{\mathrm{LCD, 16}}=\left[
\begin{array}{cccccccccccccccccccccccccccccccccc}
9 & 6 & 6 & 9 & 6 & 9 & 9 & 6 & 6 & 9 & 9 & 6 & 9 & 6 & 6 & 9 & 6 & 9 & 9 & 6 & 9 & 6 & 6 & 9 & 9 & 6 & 6 & 9 & 6 & 9 & 9 & 7 & 5 & C \\
5 & 5 & 5 & 5 & 5 & 5 & 5 & 5 & 5 & 5 & 5 & 5 & 5 & 5 & 5 & 5 & A & A & A & A & A & A & A & A & A & A & A & A & A & A & A & A & 6 & 8 \\
3 & 3 & 3 & 3 & 3 & 3 & 3 & 3 & 3 & 3 & 3 & 3 & 3 & 3 & 3 & 3 & C & C & C & C & C & C & C & C & C & C & C & C & C & C & C & C & B & 8 \\
0 & F & 0 & F & 0 & F & 0 & F & 0 & F & 0 & F & 0 & F & 0 & F & F & 0 & F & 0 & F & 0 & F & 0 & F & 0 & F & 0 & F & 0 & F & 1 & 6 & C \\
0 & 0 & F & F & 0 & 0 & F & F & 0 & 0 & F & F & 0 & 0 & F & F & F & F & 0 & 0 & F & F & 0 & 0 & F & F & 0 & 0 & F & F & 0 & 1 & D & 8 \\
0 & 0 & 0 & 0 & F & F & F & F & 0 & 0 & 0 & 0 & F & F & F & F & F & F & F & F & 0 & 0 & 0 & 0 & F & F & F & F & 0 & 0 & 0 & 0 & 8 & 0 \\
0 & 0 & 0 & 0 & 0 & 0 & 0 & 0 & F & F & F & F & F & F & F & F & F & F & F & F & F & F & F & F & 0 & 0 & 0 & 0 & 0 & 0 & 0 & 1 & 1 & 8 \\
5 & 5 & 5 & 5 & 5 & 5 & 5 & 5 & 5 & 5 & 5 & 5 & 5 & 5 & 5 & 5 & 5 & 5 & 5 & 5 & 5 & 5 & 5 & 5 & 5 & 5 & 5 & 5 & 5 & 5 & 5 & 4 & C & 8
\end{array}
\right]_{16}.
\]
}

\item An LCD $[138,8,66]$ code obtained from a $[136,8,66]$ code, whose generator matrix is
{\tiny
\[
G_{\mathrm{LCD, 17}}=\left[
\begin{array}{ccccccccccccccccccccccccccccccccccc}
8 & 0 & F & D & F & 3 & D & 7 & 0 & D & D & 3 & 1 & 5 & 8 & 2 & F & 1 & D & B & 2 & 5 & 2 & 1 & 3 & 9 & 6 & 8 & 8 & C & D & 5 & F & F & 8 \\
7 & F & 0 & 2 & 0 & C & 2 & 8 & F & 2 & 2 & C & E & A & 7 & D & 0 & E & 2 & 4 & D & A & D & E & C & 6 & 9 & 7 & 7 & 3 & 2 & B & F & E & 4 \\
4 & 0 & 7 & E & F & 9 & E & B & 8 & 6 & E & 9 & 8 & A & C & 1 & 7 & 8 & E & D & 9 & 2 & 9 & 0 & 9 & C & B & 4 & 4 & 6 & 6 & B & 0 & 3 & 8 \\
2 & 0 & C & 2 & 8 & F & 2 & 2 & C & E & A & 7 & D & 0 & E & 2 & 4 & D & A & D & E & C & 6 & 9 & 7 & 7 & 3 & 2 & A & F & E & 0 & 8 & 2 & 8 \\
1 & 0 & 6 & 1 & 4 & 7 & 9 & 1 & 6 & 7 & 5 & 3 & E & 8 & 7 & 1 & 2 & 6 & D & 6 & F & 6 & 3 & 4 & B & B & 9 & 9 & 5 & 7 & F & 0 & 4 & 2 & 0 \\
0 & 8 & 3 & 0 & A & 3 & C & 8 & B & 3 & A & 9 & F & 4 & 3 & 8 & 9 & 3 & 6 & B & 7 & B & 1 & A & 5 & D & C & C & A & B & F & 8 & 2 & 2 & 0 \\
0 & 4 & 1 & 8 & 5 & 1 & E & 4 & 5 & 9 & D & 4 & F & A & 1 & C & 4 & 9 & B & 5 & B & D & 8 & D & 2 & E & E & 6 & 5 & 5 & F & C & 1 & 2 & 0 \\
0 & 2 & 0 & C & 2 & 8 & F & 2 & 2 & C & E & A & 7 & D & 0 & E & 2 & 4 & D & A & D & E & C & 6 & 9 & 7 & 7 & 3 & 2 & A & F & E & 0 & A & 0
\end{array}
\right]_{16}.
\]
}

\item An LCD $[142,8,68]$ code obtained from a $[139,8,67]$ code, whose generator matrix is
{\tiny
\[
G_{\mathrm{LCD, 18}}=\left[
\begin{array}{cccccccccccccccccccccccccccccccccccc}
8 & 1 & 0 & 6 & 1 & 4 & 7 & 9 & 1 & 6 & 7 & 5 & 3 & E & 8 & 7 & 1 & 2 & 6 & D & 6 & F & 6 & 3 & 4 & B & B & 9 & 9 & 5 & 7 & F & F & E & 1 & 0 \\
7 & 1 & 2 & 6 & D & 6 & F & 6 & 3 & 4 & B & B & 9 & 9 & 5 & 7 & F & 0 & 2 & 0 & C & 2 & 8 & F & 2 & 2 & C & E & A & 7 & D & 1 & E & 1 & E & 8 \\
0 & E & 2 & 4 & D & A & D & E & C & 6 & 9 & 7 & 7 & 3 & 2 & A & F & E & 0 & 4 & 1 & 8 & 5 & 1 & E & 4 & 5 & 9 & D & 4 & F & A & 1 & E & 0 & C \\
8 & 0 & F & D & F & 3 & D & 7 & 0 & D & D & 3 & 1 & 5 & 8 & 2 & F & 1 & D & B & 2 & 5 & 2 & 1 & 3 & 9 & 6 & 8 & 8 & C & D & 5 & F & F & E & 0 \\
4 & 0 & 7 & E & F & 9 & E & B & 8 & 6 & E & 9 & 8 & A & C & 1 & 7 & 8 & E & D & 9 & 2 & 9 & 0 & 9 & C & B & 4 & 4 & 6 & 6 & B & 2 & 0 & C & 4 \\
2 & 0 & C & 2 & 8 & F & 2 & 2 & C & E & A & 7 & D & 0 & E & 2 & 4 & D & A & D & E & C & 6 & 9 & 7 & 7 & 3 & 2 & A & F & E & 0 & A & 0 & A & 0 \\
0 & 8 & 3 & 0 & A & 3 & C & 8 & B & 3 & A & 9 & F & 4 & 3 & 8 & 9 & 3 & 6 & B & 7 & B & 1 & A & 5 & D & C & C & A & B & F & 8 & 1 & 2 & C & 0 \\
0 & 4 & 1 & 8 & 5 & 1 & E & 4 & 5 & 9 & D & 4 & F & A & 1 & C & 4 & 9 & B & 5 & B & D & 8 & D & 2 & E & E & 6 & 5 & 5 & F & C & 0 & A & A & 0
\end{array}
\right]_{16}.
\]
}

\item An LCD $[201,8,98]$ code obtained from a $[199,8,98]$ code, whose generator matrix is
{\tiny
\[
G_{\mathrm{LCD, 19}}=\left[
\begin{array}{ccccccccccccccccccccccccccccccccccccccccccccccccccc}
8 & 0 & F & D & F & 3 & D & 7 & 0 & D & D & 3 & 1 & 5 & 8 & 2 & F & 1 & D & B & 2 & 5 & 2 & 1 & 3 & 9 & 6 & 8 & 8 & C & D & 5 & F & F & F & F & F & F & F & F & F & F & F & F & F & F & F & E & 0 & 3 & 0 \\
7 & F & 0 & 2 & 0 & C & 2 & 8 & F & 2 & 2 & C & E & A & 7 & D & 0 & E & 2 & 4 & D & A & D & E & C & 6 & 9 & 7 & 7 & 3 & 2 & B & F & F & F & F & F & F & F & F & F & F & F & F & F & F & F & E & 0 & 0 & 8 \\
4 & 0 & 7 & E & F & 9 & E & B & 8 & 6 & E & 9 & 8 & A & C & 1 & 7 & 8 & E & D & 9 & 2 & 9 & 0 & 9 & C & B & 4 & 4 & 6 & 6 & B & 0 & 3 & 7 & 1 & D & 3 & 8 & 6 & 7 & 5 & A & 2 & B & 9 & 2 & F & F & A & 0 \\
2 & 0 & C & 2 & 8 & F & 2 & 2 & C & E & A & 7 & D & 0 & E & 2 & 4 & D & A & D & E & C & 6 & 9 & 7 & 7 & 3 & 2 & A & F & E & 0 & 8 & 2 & C & 9 & 3 & A & 4 & 5 & 4 & F & 7 & 3 & E & 5 & B & 9 & 0 & 4 & 0 \\
1 & 0 & 6 & 1 & 4 & 7 & 9 & 1 & 6 & 7 & 5 & 3 & E & 8 & 7 & 1 & 2 & 6 & D & 6 & F & 6 & 3 & 4 & B & B & 9 & 9 & 5 & 7 & F & 0 & 4 & 1 & E & C & D & C & 6 & 1 & B & D & 2 & 8 & E & A & 5 & C & 8 & 4 & 0 \\
0 & 8 & 3 & 0 & A & 3 & C & 8 & B & 3 & A & 9 & F & 4 & 3 & 8 & 9 & 3 & 6 & B & 7 & B & 1 & A & 5 & D & C & C & A & B & F & 8 & 2 & 0 & 4 & 6 & 8 & F & 7 & 2 & D & E & 9 & 5 & 3 & 9 & B & C & 4 & 4 & 0 \\
0 & 4 & 1 & 8 & 5 & 1 & E & 4 & 5 & 9 & D & 4 & F & A & 1 & C & 4 & 9 & B & 5 & B & D & 8 & D & 2 & E & E & 6 & 5 & 5 & F & C & 1 & 0 & 1 & B & E & E & 3 & 8 & 6 & 7 & 9 & A & D & 4 & 5 & E & 2 & 4 & 0 \\
0 & 2 & 0 & C & 2 & 8 & F & 2 & 2 & C & E & A & 7 & D & 0 & E & 2 & 4 & D & A & D & E & C & 6 & 9 & 7 & 7 & 3 & 2 & A & F & E & 0 & 8 & 2 & 5 & 1 & F & 9 & D & 3 & B & 4 & C & 7 & 6 & A & E & 1 & 4 & 0
\end{array}
\right]_{16}.
\]
}
\end{itemize}
\endgroup
Among these 19 codes, the $[132,6,65]$ and $[195,6,97]$ codes are equivalent to codes in the MAGMA BKLC database, whereas the remaining codes are inequivalent to all BKLC codes. Also, the $[69,6,33]$, $[162,6,80]$ and $[225,6,112]$ codes are equivalent to the codes in~\cite{LL-2024}. For the remaining 14 codes, we could not find any equivalent codes in the previous studies. Moreover, these codes attain the upper bound, and hence determine the exact values $d_{LCD}(n, k)$ for the corresponding parameters $n$ and $k$.

Therefore, we obtain the following theorem.

\begin{theorem}\label{thm:new-optimal-lcd-codes}
There exist new binary optimal LCD codes with the following parameters:
\[
\begin{gathered}
\relax [55,7,25],\ [57,7,26],\ [69,7,32],\ [72,7,34],\ [73,7,34],\\
[85,7,40],\ [134,7,65],\ [136,7,66],\ [199,7,98],\\
[57,8,25],\ [134,8,64],\ [138,8,66],\ [142,8,68],\ [201,8,98].
\end{gathered}
\]
\end{theorem}

The results for codes of dimensions $6\leq k\leq 10$ are presented in Tables~\ref{tab:lcd_embedding_k6}, \ref{tab:lcd_embedding_k7}, \ref{tab:lcd_embedding_k8}, \ref{tab:lcd_embedding_k9} and \ref{tab:lcd_embedding_k10}. For each parameter set, we performed Experiment 1 twice and Experiment 2 twice, and recorded the best outcome among the four runs.

In the tables, $n$ denotes the length of the code obtained after applying a shortest LCD embedding. The column ``Bound'' gives the upper bound on $d_{LCD}(n,k)$ for $6\leq k\leq 10$, derived in Subsection~\ref{experiments-bound-LCD}. The column ``Our result'' denotes the largest minimum distance obtained by our method for the corresponding parameters. An asterisk $(*)$ indicates that the upper bound is attained.

The generator matrices of resulting LCD embeddings are available at~\cite{github}.

\begin{table}[p]
\centering
\rotatebox{\lcdGreedyTableRotationAngle}{%
\begin{minipage}{0.92\textheight}
\centering
\caption{Highest minimum distances obtained from shortest LCD embeddings for $k=6$}
\label{tab:lcd_embedding_k6}
\scriptsize
\setlength{\tabcolsep}{2.2pt}
\renewcommand{\arraystretch}{1.02}
\resizebox{\linewidth}{!}{%
\begin{tabular}{c|r|c||c|r|c||c|r|c||c|r|c||c|r|c||c|r|c}
\Xhline{2\arrayrulewidth}
$n$ & Bound & Our result & $n$ & Bound & Our result & $n$ & Bound & Our result & $n$ & Bound & Our result & $n$ & Bound & Our result & $n$ & Bound & Our result \\
\Xhline{2\arrayrulewidth}
52 & 25 & 24 & 86 & 42 & 41 & 121 & 60 & 59 & 157 & 79 & 77 & 191 & 96 & 94 & 227 & 114 & 113\\
\hline
54 & 26 & 25 & 87 & 43 & 42 & 123 & 62 & 59 & 158 & 80 & 77 & 193 & 96 & 95 & 228 & 114 & 113\\
\hline
55 & 27 & 26 & 88 & 44 & 42 & 125 & 63 & 61 & 159 & 80 & 78 & 195 & 97 & 97* & 229 & 115 & 114\\
\hline
56 & 27 & 25 & 89 & 44 & 41 & 126 & 63 & 61 & 160 & 80 & 78 & 196 & 98 & 97 & 230 & 116 & 113\\
\hline
57 & 28 & 25 & 92 & 46 & 44 & 127 & 64 & 62 & 162 & 80 & 80* & 197 & 98 & 97 & 231 & 116 & 113\\
\hline
58 & 28 & 27 & 94 & 47 & 45 & 128 & 64 & 62 & 164 & 82 & 81 & 198 & 99 & 97 & 232 & 116 & 115\\
\hline
60 & 29 & 28 & 95 & 47 & 45 & 130 & 64 & 63 & 165 & 82 & 81 & 199 & 100 & 98 & 234 & 118 & 116\\
\hline
62 & 30 & 29 & 96 & 48 & 46 & 132 & 65 & 65* & 166 & 83 & 82 & 200 & 100 & 97 & 236 & 119 & 117\\
\hline
63 & 31 & 29 & 97 & 48 & 46 & 133 & 66 & 65 & 167 & 84 & 82 & 203 & 102 & 100 & 237 & 120 & 117\\
\hline
64 & 31 & 30 & 99 & 48 & 47 & 134 & 66 & 65 & 168 & 84 & 81 & 205 & 103 & 102 & 238 & 120 & 119\\
\hline
65 & 32 & 30 & 101 & 50 & 49 & 135 & 67 & 65 & 171 & 86 & 84 & 206 & 104 & 101 & 239 & 120 & 119\\
\hline
67 & 32 & 31 & 102 & 50 & 49 & 136 & 68 & 65 & 173 & 87 & 85 & 207 & 104 & 103 & 241 & 121 & 120\\
\hline
69 & 33 & 33* & 103 & 51 & 50 & 137 & 68 & 65 & 174 & 88 & 85 & 208 & 104 & 103 & 243 & 122 & 121\\
\hline
70 & 34 & 33 & 104 & 52 & 50 & 140 & 70 & 68 & 175 & 88 & 87 & 210 & 105 & 104 & 244 & 123 & 122\\
\hline
71 & 34 & 33 & 105 & 52 & 49 & 142 & 71 & 69 & 176 & 88 & 87 & 211 & 106 & 105 & 245 & 124 & 121\\
\hline
72 & 35 & 34 & 108 & 54 & 52 & 143 & 72 & 69 & 178 & 89 & 88 & 212 & 106 & 105 & 246 & 124 & 121\\
\hline
73 & 36 & 34 & 110 & 55 & 53 & 144 & 72 & 70 & 180 & 90 & 89 & 213 & 107 & 106 & 247 & 124 & 122\\
\hline
74 & 36 & 33 & 111 & 55 & 53 & 145 & 72 & 70 & 181 & 91 & 90 & 214 & 108 & 105 & 249 & 126 & 124\\
\hline
77 & 38 & 37 & 112 & 56 & 54 & 147 & 73 & 72 & 182 & 92 & 89 & 216 & 108 & 106 & 251 & 127 & 125\\
\hline
79 & 39 & 38 & 113 & 56 & 54 & 149 & 74 & 73 & 183 & 92 & 89 & 218 & 110 & 108 & 252 & 128 & 125\\
\hline
80 & 39 & 38 & 115 & 57 & 55 & 150 & 75 & 74 & 184 & 92 & 91 & 220 & 111 & 109 & 253 & 128 & 126\\
\hline
81 & 40 & 39 & 117 & 58 & 57 & 151 & 76 & 73 & 186 & 94 & 92 & 221 & 112 & 109 & 256 & 128 & 127\\
\hline
82 & 40 & 39 & 118 & 59 & 57 & 152 & 76 & 75 & 188 & 95 & 93 & 222 & 112 & 110 &  &  & \\
\hline
84 & 41 & 40 & 119 & 60 & 57 & 153 & 76 & 75 & 189 & 96 & 93 & 223 & 112 & 110 &  &  & \\
\hline
85 & 42 & 41 & 120 & 60 & 59 & 155 & 78 & 76 & 190 & 96 & 94 & 225 & 112 & 112* &  &  & \\
\Xhline{2\arrayrulewidth}
\end{tabular}%
}
\end{minipage}%
}
\end{table}

\begin{table}[p]
\centering
\rotatebox{\lcdGreedyTableRotationAngle}{%
\begin{minipage}{0.92\textheight}
\centering
\caption{Highest minimum distances obtained from shortest LCD embeddings for $k=7$}
\label{tab:lcd_embedding_k7}
\scriptsize
\setlength{\tabcolsep}{2.2pt}
\renewcommand{\arraystretch}{1.02}
\resizebox{\linewidth}{!}{%
\begin{tabular}{c|r|c||c|r|c||c|r|c||c|r|c||c|r|c||c|r|c}
\Xhline{2\arrayrulewidth}
$n$ & Bound & Our result & $n$ & Bound & Our result & $n$ & Bound & Our result & $n$ & Bound & Our result & $n$ & Bound & Our result & $n$ & Bound & Our result \\
\Xhline{2\arrayrulewidth}
51 & 24 & 22 & 85 & 40 & 40* & 121 & 60 & 57 & 155 & 76 & 75 & 192 & 96 & 93 & 226 & 112 & 111\\
\hline
52 & 24 & 23 & 87 & 42 & 41 & 122 & 60 & 59 & 158 & 78 & 76 & 193 & 96 & 94 & 228 & 113 & 111\\
\hline
53 & 24 & 23 & 88 & 42 & 41 & 123 & 61 & 59 & 160 & 80 & 78 & 194 & 96 & 94 & 230 & 114 & 113\\
\hline
55 & 25 & 25* & 89 & 43 & 41 & 125 & 62 & 60 & 161 & 80 & 78 & 196 & 96 & 95 & 231 & 115 & 113\\
\hline
56 & 26 & 25 & 90 & 44 & 42 & 127 & 63 & 61 & 162 & 80 & 79 & 198 & 98 & 97 & 232 & 116 & 114\\
\hline
57 & 26 & 26* & 92 & 45 & 43 & 128 & 64 & 61 & 163 & 80 & 79 & 199 & 98 & 98* & 233 & 116 & 113\\
\hline
58 & 27 & 26 & 94 & 46 & 44 & 129 & 64 & 62 & 165 & 81 & 80 & 200 & 99 & 98 & 234 & 116 & 115\\
\hline
59 & 28 & 25 & 96 & 47 & 45 & 130 & 64 & 62 & 167 & 82 & 81 & 201 & 100 & 97 & 235 & 117 & 115\\
\hline
60 & 28 & 25 & 98 & 48 & 47 & 132 & 64 & 63 & 168 & 83 & 81 & 202 & 100 & 97 & 237 & 118 & 116\\
\hline
61 & 29 & 26 & 99 & 48 & 47 & 134 & 65 & 65* & 169 & 84 & 82 & 203 & 100 & 97 & 239 & 120 & 118\\
\hline
64 & 31 & 29 & 101 & 49 & 48 & 135 & 66 & 65 & 170 & 84 & 81 & 206 & 102 & 101 & 240 & 120 & 117\\
\hline
65 & 31 & 30 & 103 & 50 & 49 & 136 & 66 & 66* & 171 & 84 & 81 & 208 & 104 & 102 & 241 & 120 & 119\\
\hline
66 & 32 & 30 & 104 & 51 & 49 & 137 & 67 & 66 & 172 & 85 & 83 & 209 & 104 & 102 & 242 & 120 & 119\\
\hline
67 & 32 & 31 & 105 & 52 & 50 & 138 & 68 & 65 & 174 & 86 & 84 & 210 & 104 & 103 & 244 & 122 & 119\\
\hline
69 & 32 & 32* & 106 & 52 & 49 & 139 & 68 & 65 & 176 & 88 & 86 & 211 & 104 & 103 & 246 & 123 & 121\\
\hline
71 & 34 & 33 & 107 & 52 & 49 & 140 & 68 & 65 & 177 & 88 & 86 & 213 & 106 & 104 & 247 & 124 & 121\\
\hline
72 & 34 & 34* & 108 & 53 & 52 & 142 & 70 & 68 & 178 & 88 & 87 & 215 & 107 & 105 & 248 & 124 & 121\\
\hline
73 & 34 & 34* & 110 & 54 & 52 & 143 & 70 & 69 & 179 & 88 & 86 & 216 & 108 & 105 & 249 & 124 & 123\\
\hline
74 & 35 & 34 & 112 & 55 & 53 & 144 & 71 & 70 & 181 & 90 & 88 & 217 & 108 & 105 & 250 & 125 & 123\\
\hline
75 & 36 & 33 & 113 & 56 & 53 & 145 & 72 & 70 & 183 & 91 & 90 & 218 & 108 & 107 & 252 & 126 & 123\\
\hline
76 & 36 & 34 & 114 & 56 & 55 & 147 & 72 & 70 & 184 & 92 & 90 & 219 & 109 & 107 & 254 & 128 & 125\\
\hline
78 & 37 & 36 & 115 & 56 & 55 & 150 & 74 & 72 & 185 & 92 & 89 & 221 & 110 & 108 & 255 & 128 & 125\\
\hline
80 & 38 & 37 & 117 & 58 & 56 & 152 & 75 & 73 & 186 & 92 & 91 & 223 & 112 & 109 & 256 & 128 & 126\\
\hline
82 & 40 & 38 & 119 & 59 & 57 & 153 & 76 & 74 & 189 & 94 & 92 & 224 & 112 & 109 &  &  & \\
\hline
83 & 40 & 38 & 120 & 60 & 58 & 154 & 76 & 74 & 191 & 96 & 93 & 225 & 112 & 111 &  &  & \\
\Xhline{2\arrayrulewidth}
\end{tabular}%
}
\end{minipage}%
}
\end{table}

\begin{table}[p]
\centering
\rotatebox{\lcdGreedyTableRotationAngle}{%
\begin{minipage}{0.92\textheight}
\centering
\caption{Highest minimum distances obtained from shortest LCD embeddings for $k=8$}
\label{tab:lcd_embedding_k8}
\scriptsize
\setlength{\tabcolsep}{2.2pt}
\renewcommand{\arraystretch}{1.02}
\resizebox{\linewidth}{!}{%
\begin{tabular}{c|r|c||c|r|c||c|r|c||c|r|c||c|r|c||c|r|c}
\Xhline{2\arrayrulewidth}
$n$ & Bound & Our result & $n$ & Bound & Our result & $n$ & Bound & Our result & $n$ & Bound & Our result & $n$ & Bound & Our result & $n$ & Bound & Our result \\
\Xhline{2\arrayrulewidth}
51 & 23 & 21 & 85 & 40 & 38 & 119 & 57 & 55 & 157 & 76 & 74 & 189 & 94 & 91 & 226 & 112 & 109\\
\hline
52 & 24 & 22 & 86 & 40 & 38 & 121 & 58 & 57 & 158 & 77 & 74 & 191 & 95 & 92 & 227 & 112 & 111\\
\hline
53 & 24 & 22 & 87 & 40 & 39 & 122 & 58 & 57 & 159 & 78 & 76 & 193 & 96 & 93 & 228 & 112 & 111\\
\hline
54 & 24 & 23 & 88 & 41 & 39 & 123 & 59 & 57 & 161 & 79 & 77 & 195 & 96 & 95 & 230 & 114 & 112\\
\hline
55 & 24 & 23 & 90 & 42 & 41 & 124 & 60 & 58 & 163 & 80 & 78 & 196 & 96 & 95 & 232 & 115 & 113\\
\hline
57 & 25 & 25* & 92 & 44 & 41 & 125 & 61 & 58 & 164 & 80 & 78 & 198 & 97 & 96 & 233 & 116 & 113\\
\hline
58 & 26 & 25 & 93 & 44 & 42 & 126 & 62 & 57 & 165 & 80 & 79 & 199 & 98 & 97 & 234 & 116 & 114\\
\hline
59 & 26 & 25 & 94 & 44 & 42 & 128 & 63 & 61 & 166 & 80 & 79 & 200 & 98 & 97 & 235 & 116 & 113\\
\hline
60 & 27 & 26 & 95 & 45 & 43 & 130 & 64 & 62 & 168 & 82 & 80 & 201 & 98 & 98* & 236 & 117 & 115\\
\hline
61 & 28 & 25 & 96 & 46 & 44 & 132 & 64 & 63 & 170 & 83 & 81 & 202 & 99 & 98 & 237 & 118 & 115\\
\hline
62 & 28 & 26 & 98 & 47 & 45 & 134 & 64 & 64* & 171 & 84 & 81 & 203 & 100 & 97 & 239 & 119 & 117\\
\hline
63 & 28 & 26 & 100 & 48 & 46 & 136 & 66 & 65 & 172 & 84 & 82 & 204 & 100 & 97 & 241 & 120 & 118\\
\hline
64 & 29 & 26 & 101 & 48 & 46 & 137 & 66 & 65 & 173 & 84 & 82 & 205 & 100 & 97 & 242 & 120 & 117\\
\hline
65 & 30 & 27 & 102 & 48 & 47 & 138 & 66 & 66* & 174 & 85 & 83 & 207 & 102 & 100 & 243 & 120 & 119\\
\hline
67 & 31 & 29 & 103 & 48 & 47 & 139 & 67 & 66 & 175 & 86 & 84 & 209 & 103 & 101 & 244 & 121 & 119\\
\hline
69 & 32 & 30 & 105 & 50 & 48 & 140 & 68 & 67 & 177 & 87 & 85 & 211 & 104 & 102 & 246 & 122 & 120\\
\hline
70 & 32 & 30 & 107 & 51 & 49 & 141 & 68 & 65 & 179 & 88 & 86 & 212 & 104 & 102 & 247 & 123 & 120\\
\hline
71 & 32 & 31 & 108 & 52 & 49 & 142 & 68 & 68* & 180 & 88 & 86 & 213 & 105 & 103 & 248 & 124 & 121\\
\hline
72 & 32 & 31 & 109 & 52 & 50 & 143 & 69 & 68 & 181 & 88 & 87 & 215 & 106 & 104 & 249 & 124 & 122\\
\hline
74 & 34 & 33 & 110 & 52 & 50 & 144 & 70 & 69 & 182 & 89 & 87 & 216 & 107 & 105 & 250 & 124 & 121\\
\hline
76 & 35 & 34 & 111 & 53 & 49 & 145 & 70 & 68 & 183 & 90 & 89 & 217 & 108 & 105 & 251 & 125 & 123\\
\hline
77 & 36 & 33 & 112 & 54 & 52 & 146 & 71 & 70 & 184 & 90 & 88 & 218 & 108 & 106 & 252 & 126 & 123\\
\hline
78 & 36 & 34 & 114 & 55 & 53 & 151 & 73 & 72 & 185 & 91 & 90 & 220 & 109 & 107 & 254 & 127 & 124\\
\hline
79 & 36 & 35 & 116 & 56 & 54 & 153 & 74 & 73 & 186 & 92 & 89 & 221 & 110 & 107 & 256 & 128 & 125\\
\hline
80 & 37 & 33 & 117 & 56 & 54 & 155 & 76 & 74 & 187 & 92 & 89 & 223 & 111 & 108 &  &  & \\
\hline
83 & 39 & 37 & 118 & 56 & 55 & 156 & 76 & 74 & 188 & 93 & 89 & 225 & 112 & 110 &  &  & \\
\Xhline{2\arrayrulewidth}
\end{tabular}%
}
\end{minipage}%
}
\end{table}

\begin{table}[p]
\centering
\rotatebox{\lcdGreedyTableRotationAngle}{%
\begin{minipage}{0.92\textheight}
\centering
\caption{Highest minimum distances obtained from shortest LCD embeddings for $k=9$}
\label{tab:lcd_embedding_k9}
\scriptsize
\setlength{\tabcolsep}{2.2pt}
\renewcommand{\arraystretch}{1.02}
\resizebox{\linewidth}{!}{%
\begin{tabular}{c|r|c||c|r|c||c|r|c||c|r|c||c|r|c||c|r|c}
\Xhline{2\arrayrulewidth}
$n$ & Bound & Our result & $n$ & Bound & Our result & $n$ & Bound & Our result & $n$ & Bound & Our result & $n$ & Bound & Our result & $n$ & Bound & Our result \\
\Xhline{2\arrayrulewidth}
53 & 23 & 21 & 89 & 41 & 38 & 123 & 58 & 55 & 158 & 76 & 74 & 190 & 92 & 89 & 225 & 110 & 106\\
\hline
55 & 24 & 22 & 90 & 42 & 39 & 125 & 60 & 57 & 159 & 76 & 74 & 191 & 92 & 90 & 226 & 110 & 107\\
\hline
56 & 24 & 22 & 91 & 42 & 39 & 127 & 60 & 58 & 160 & 77 & 74 & 192 & 93 & 89 & 228 & 112 & 109\\
\hline
57 & 24 & 22 & 93 & 44 & 41 & 128 & 61 & 58 & 161 & 78 & 74 & 193 & 94 & 91 & 230 & 112 & 110\\
\hline
58 & 25 & 23 & 95 & 44 & 42 & 129 & 62 & 58 & 162 & 78 & 74 & 194 & 94 & 91 & 231 & 112 & 110\\
\hline
59 & 26 & 23 & 96 & 44 & 42 & 130 & 62 & 59 & 163 & 79 & 75 & 196 & 96 & 93 & 232 & 113 & 111\\
\hline
61 & 27 & 25 & 97 & 45 & 42 & 131 & 62 & 60 & 164 & 80 & 76 & 198 & 96 & 94 & 233 & 114 & 111\\
\hline
63 & 28 & 26 & 98 & 46 & 42 & 133 & 64 & 61 & 166 & 80 & 77 & 199 & 96 & 93 & 235 & 115 & 112\\
\hline
64 & 28 & 26 & 99 & 46 & 43 & 135 & 64 & 62 & 167 & 80 & 78 & 200 & 96 & 95 & 237 & 116 & 113\\
\hline
65 & 29 & 26 & 100 & 47 & 43 & 136 & 64 & 62 & 168 & 80 & 78 & 203 & 98 & 96 & 238 & 116 & 113\\
\hline
66 & 30 & 27 & 102 & 48 & 45 & 137 & 64 & 63 & 169 & 81 & 79 & 205 & 100 & 97 & 239 & 116 & 113\\
\hline
67 & 30 & 28 & 104 & 48 & 46 & 138 & 65 & 64 & 172 & 83 & 81 & 206 & 100 & 98 & 240 & 117 & 114\\
\hline
69 & 31 & 29 & 105 & 48 & 46 & 140 & 66 & 65 & 174 & 84 & 81 & 207 & 100 & 97 & 241 & 118 & 113\\
\hline
71 & 32 & 30 & 106 & 49 & 47 & 141 & 67 & 66 & 175 & 84 & 81 & 208 & 101 & 97 & 242 & 118 & 115\\
\hline
72 & 32 & 30 & 107 & 50 & 47 & 142 & 68 & 65 & 176 & 85 & 82 & 209 & 102 & 97 & 243 & 119 & 118\\
\hline
73 & 32 & 30 & 109 & 51 & 49 & 143 & 68 & 65 & 177 & 86 & 83 & 210 & 102 & 97 & 245 & 120 & 117\\
\hline
74 & 33 & 31 & 111 & 52 & 50 & 144 & 68 & 66 & 178 & 86 & 82 & 212 & 104 & 100 & 247 & 121 & 118\\
\hline
75 & 34 & 31 & 112 & 52 & 50 & 145 & 69 & 66 & 179 & 87 & 84 & 214 & 104 & 101 & 248 & 122 & 118\\
\hline
77 & 35 & 33 & 113 & 53 & 50 & 146 & 70 & 66 & 180 & 88 & 84 & 216 & 105 & 102 & 250 & 122 & 119\\
\hline
79 & 36 & 34 & 114 & 54 & 51 & 147 & 70 & 66 & 182 & 88 & 85 & 217 & 106 & 102 & 252 & 124 & 121\\
\hline
80 & 36 & 34 & 115 & 54 & 51 & 150 & 72 & 69 & 183 & 88 & 86 & 218 & 106 & 102 & 254 & 126 & 122\\
\hline
81 & 37 & 33 & 117 & 56 & 53 & 152 & 72 & 70 & 184 & 89 & 86 & 219 & 107 & 103 & 255 & 127 & 122\\
\hline
82 & 38 & 33 & 119 & 56 & 54 & 153 & 73 & 70 & 185 & 90 & 87 & 221 & 108 & 105 & 256 & 128 & 122\\
\hline
84 & 39 & 36 & 120 & 56 & 54 & 154 & 74 & 71 & 186 & 90 & 88 & 222 & 108 & 106 &  &  & \\
\hline
86 & 40 & 37 & 121 & 57 & 54 & 155 & 74 & 72 & 188 & 91 & 88 & 223 & 109 & 106 &  &  & \\
\hline
88 & 40 & 38 & 122 & 58 & 55 & 157 & 76 & 73 & 189 & 92 & 89 & 224 & 110 & 106 &  &  & \\
\Xhline{2\arrayrulewidth}
\end{tabular}%
}
\end{minipage}%
}
\end{table}

\begin{table}[p]
\centering
\rotatebox{\lcdGreedyTableRotationAngle}{%
\begin{minipage}{0.92\textheight}
\centering
\caption{Highest minimum distances obtained from shortest LCD embeddings for $k=10$}
\label{tab:lcd_embedding_k10}
\scriptsize
\setlength{\tabcolsep}{2.2pt}
\renewcommand{\arraystretch}{1.02}
\resizebox{\linewidth}{!}{%
\begin{tabular}{c|r|c||c|r|c||c|r|c||c|r|c||c|r|c||c|r|c}
\Xhline{2\arrayrulewidth}
$n$ & Bound & Our result & $n$ & Bound & Our result & $n$ & Bound & Our result & $n$ & Bound & Our result & $n$ & Bound & Our result & $n$ & Bound & Our result \\
\Xhline{2\arrayrulewidth}
51 & 21 & 19 & 85 & 38 & 34 & 117 & 55 & 51 & 150 & 72 & 66 & 186 & 89 & 86 & 222 & 108 & 103\\
\hline
53 & 22 & 18 & 86 & 39 & 34 & 118 & 56 & 52 & 151 & 72 & 68 & 187 & 90 & 86 & 224 & 108 & 105\\
\hline
54 & 23 & 20 & 87 & 40 & 36 & 120 & 56 & 53 & 153 & 72 & 69 & 188 & 90 & 86 & 226 & 110 & 106\\
\hline
56 & 24 & 21 & 88 & 40 & 37 & 122 & 57 & 54 & 155 & 74 & 70 & 189 & 91 & 87 & 227 & 110 & 106\\
\hline
58 & 24 & 23 & 89 & 40 & 37 & 123 & 58 & 54 & 156 & 74 & 70 & 191 & 92 & 89 & 228 & 111 & 107\\
\hline
59 & 25 & 23 & 90 & 40 & 38 & 124 & 58 & 55 & 157 & 75 & 71 & 192 & 92 & 89 & 229 & 112 & 107\\
\hline
60 & 26 & 22 & 91 & 41 & 38 & 125 & 58 & 55 & 158 & 76 & 71 & 194 & 93 & 90 & 231 & 112 & 109\\
\hline
61 & 26 & 22 & 92 & 42 & 39 & 127 & 60 & 57 & 160 & 76 & 73 & 196 & 94 & 91 & 233 & 112 & 110\\
\hline
62 & 27 & 24 & 93 & 42 & 40 & 129 & 60 & 57 & 162 & 78 & 74 & 199 & 96 & 93 & 234 & 113 & 109\\
\hline
64 & 28 & 25 & 95 & 44 & 41 & 130 & 61 & 58 & 163 & 78 & 74 & 201 & 96 & 94 & 235 & 114 & 111\\
\hline
66 & 28 & 26 & 97 & 44 & 41 & 131 & 62 & 59 & 164 & 78 & 74 & 202 & 96 & 94 & 236 & 114 & 111\\
\hline
67 & 29 & 27 & 98 & 44 & 42 & 132 & 62 & 58 & 165 & 79 & 74 & 203 & 97 & 95 & 238 & 116 & 112\\
\hline
68 & 30 & 26 & 99 & 45 & 42 & 133 & 63 & 58 & 166 & 80 & 75 & 204 & 98 & 95 & 240 & 116 & 113\\
\hline
69 & 30 & 26 & 100 & 46 & 42 & 134 & 64 & 60 & 168 & 80 & 77 & 206 & 99 & 96 & 241 & 117 & 113\\
\hline
70 & 31 & 27 & 101 & 46 & 43 & 136 & 64 & 61 & 169 & 80 & 77 & 208 & 100 & 97 & 242 & 118 & 114\\
\hline
72 & 32 & 29 & 102 & 47 & 44 & 138 & 64 & 62 & 171 & 82 & 78 & 209 & 100 & 97 & 243 & 118 & 114\\
\hline
74 & 32 & 30 & 104 & 48 & 45 & 139 & 65 & 62 & 173 & 83 & 79 & 210 & 101 & 97 & 244 & 119 & 113\\
\hline
75 & 33 & 30 & 106 & 48 & 46 & 140 & 66 & 62 & 174 & 84 & 79 & 211 & 102 & 97 & 245 & 120 & 113\\
\hline
76 & 34 & 31 & 107 & 49 & 46 & 141 & 66 & 63 & 176 & 84 & 81 & 212 & 102 & 98 & 246 & 120 & 116\\
\hline
77 & 34 & 32 & 108 & 50 & 47 & 142 & 67 & 64 & 178 & 85 & 82 & 213 & 103 & 97 & 248 & 120 & 116\\
\hline
79 & 36 & 33 & 109 & 50 & 48 & 144 & 68 & 65 & 179 & 86 & 81 & 214 & 104 & 99 & 249 & 121 & 117\\
\hline
80 & 36 & 33 & 111 & 52 & 49 & 145 & 68 & 66 & 180 & 86 & 82 & 215 & 104 & 99 & 252 & 123 & 118\\
\hline
81 & 36 & 34 & 113 & 52 & 50 & 146 & 69 & 66 & 181 & 87 & 82 & 217 & 104 & 101 & 254 & 124 & 119\\
\hline
82 & 36 & 34 & 114 & 53 & 50 & 147 & 70 & 66 & 182 & 88 & 84 & 219 & 106 & 102 & 255 & 124 & 120\\
\hline
83 & 37 & 34 & 115 & 54 & 50 & 148 & 70 & 66 & 183 & 88 & 83 & 220 & 106 & 102 &  &  & \\
\hline
84 & 38 & 34 & 116 & 54 & 50 & 149 & 71 & 66 & 185 & 88 & 85 & 221 & 107 & 103 &  &  & \\
\Xhline{2\arrayrulewidth}
\end{tabular}%
}
\end{minipage}%
}
\end{table}

% The title of your section 5:
\section{Conclusion}
In this paper, we have constructed optimal LCD codes by applying a greedy algorithm to shortest LCD embedding methods. Since an invertible matrix together with an arbitrary matrix yields an LCD embedding of a given code, we have constructed our Greedy algorithm by using elementary row operations for an invertible matrix and entry-flip operations for an arbitrary matrix as local moves. We have obtained $19$ optimal LCD codes for lengths $51\le n\le 256$ and dimensions $6\le k\le 8$, 14 of which are new.

or dimensions $k=9$ and $k=10$, however, we did not obtain any codes attaining the corresponding bounds. Constructing optimal LCD codes in these dimensions remains a topic for future work.

\section*{Acknowledgement}
This research (J.-L. Kim) was supported in part by the BK21 FOUR (Fostering Outstanding Universities for Research) funded by the Ministry of Education (MOE, Korea), National Research Foundation of Korea (NRF) under Grant No. 4120240415042, Basic Science Research Program through the National Research Foundation of Korea (NRF) funded by the Ministry of Science and ICT under Grant No. RS-2025-24534992 and Global - Learning \& Academic research institution for Master’s·PhD students, and Postdocs(LAMP) Program of the National Research Foundation of Korea(NRF) grant funded by the Ministry of Education(No. RS-2024-00441954).

\end{document}